\documentclass[11pt]{article}
\usepackage[T1]{fontenc}
\usepackage{epsfig}
\newtheorem{theorem}{Theorem}
\newtheorem{proposition}{Proposition}
\newtheorem{lemma}{Lemma}

\newenvironment{proof}{\begin{trivlist}\item[]{\bf Proof:}\rm}{\hfill\rule{2mm}{2mm}\end{trivlist}}

\newcommand {\dfn} {\stackrel{\Delta} {=}}

\newcommand {\bone} {\mbox{\boldmath $1$}}

\newcommand {\bx} {\mbox{\boldmath $x$}}

\newcommand {\bz} {\mbox{\boldmath $z$}}

\newcommand{\calB}{{\cal B}}

\newcommand{\calS}{{\cal S}}
\newcommand{\calT}{{\cal T}}

\newcommand{\calX}{{\cal X}}

\newcommand{\bxh}{\hat{\bx}}
\newcommand{\Xh}{\hat{\mathcal X}}

\begin{document}
\title{Universal Random Coding for Successive Refinement\\
of Individual Sequences Based on Lempel--Ziv Complexity}
\author{Neri Merhav}
\date{}
\maketitle
\thispagestyle{empty}

\begin{center}
The Viterbi Faculty of Electrical and Computer Engineering\\
Technion - Israel Institute of Technology \\
Technion City, Haifa 3200003, ISRAEL \\
E--mail: {\tt merhav@technion.ac.il}\\
\end{center}
\vspace{1.5\baselineskip}
\setlength{\baselineskip}{1.5\baselineskip}

\begin{abstract}
This work extends an earlier proposed universal random-coding ensemble for sample-wise
lossy compression of individual sequences to the two-stage
(successive-refinement) setting --- an extension that is not
quite straightforward, as explained in the sequel. The construction has
three layers. First, a complete, self-contained direct/converse match
in the language of type classes and typical sequences, a two-stage analogue of Rimoldi's
classical characterization (applied in the superalphabet of blocks). 
Second, a realization of the same rates via a
random-coding scheme built from the Lempel--Ziv (LZ78) complexity rather
than the type-class uniform measure the first layer relies on, since an
LZ-based ensemble is universal across source statistics and
distortion measures, carries no block-length parameter in the
codebook itself, and is directly implementable. Third, we give a
fully type-free achievability construction (no joint types of
$\ell$-vectors appear anywhere in its search criterion) and show
it matches the same converse exactly, reaching every
point of the achievable region: 
the same converse from the first layer already binds any code, 
type-free schemes included, so no separate argument is needed to certify the match.
We also describe (without proof) how to
extend the type-free construction to any fixed
number $r>2$ of stages.
Finally, we show this paper's achievable
region is a superset of the one obtained by a companion paper,
which is based on a finite-state modeling approach to a similar two-stage problem,
extending a known single-stage domination result to two stages.
\end{abstract}

\section{Introduction}

This paper is about compressing an arbitrary individual (deterministic)
sequence of data progressively, in two stages: a first,
coarse description that satisfies a
prescribed distortion level $D_1$, followed by a second description
that adds refinement, held to its own prescribed level $D_2$ ---
possibly under a different distortion measure than the first stage's,
as detailed in Section~\ref{sec:setup}. This kind of
progressive compression is natural whenever a
receiver may need to stop listening relatively early (imagine a picture that is first
shown in low resolution and then progressively sharpened) 
and the question this paper addresses is how to do it well,
and how to know when a given scheme is already as good as possible.

A recent paper \cite{merhav233} addressed the
single-stage version of this problem: how to compress a single fixed
sequence, in one shot, as well as possible, without knowing anything
about its statistics in advance. The scheme it proposed is based on
universal random coding: each candidate reproduction codeword $\bxh$ is drawn
independently at random, using a universal probability distribution
proportional to $2^{-LZ(\bxh)}$, where $LZ(\cdot)$ is the
Lempel--Ziv (LZ78) codelength \cite{zivlempel78}
and the index of the first
candidate accurate enough is transmitted; it was shown that no coding
scheme can do substantially better than this, for the given sequence.
The aim of the present paper is to extend
the same random-coding approach to progressive compression with two
stages, and to show that the resulting scheme is, again, essentially as good as
any two-stage scheme could be; this two-stage case is this paper's
actual content, developed in three stages throughout: first entirely in the
language of type classes (Section~\ref{sec:types}), second --- via a random-coding
scheme built from LZ78 complexity (Section~\ref{sec:lz}), and third ---
separately, via a fully type-free construction matching the same
converse (Sections~\ref{sec:typefree}--\ref{sec:tf-converse}). Finally,
Section~\ref{sec:mstage} gives a description (without
a full proof)
of how the type-free mechanism
extends to a general fixed number $r>2$ of stages. A companion paper
\cite{merhav244} studies a similar two-stage compression
problem of individual sequences, but via a different paradigm:
a finite-state lossless encoder fed by a vector quantizer.
The present paper stays within the random-coding approach of
\cite{merhav233} instead, and we show (Section~\ref{sec:relation244})
that this paper's achievable region is a superset of
\cite{merhav244}'s, extending \cite{merhav233}'s own single-stage
domination of any finite-state scheme to two stages.

Extending the lossy compression scheme of \cite{merhav233} 
into two stages might sound routine once
the single-stage version is settled, but it is not --- not because
successive refinement itself harbors any new subtlety (it does not:
both difficulties below already appear in Rimoldi's own, classical
characterization), but because transplanting them to the
individual-sequence, universal setting \cite{merhav233}'s own
framework operates in takes real work --- work that a further
ambition, avoiding any commitment to a type in the search criterion,
makes harder still.

The first is on the achievability side: the coarse description
$\bxh_1$ must also serve as good side information for the second
stage, not merely satisfy the $D_1$ constraint. Rimoldi's own
fix --- commit Stage 1's search to the needed joint behavior, not
mere distortion-feasibility --- transplants with new vocabulary:
type in place of law, the universal LZ measure in place of a
type-class-uniform one (Section~\ref{sec:why-hard}). A fully
type-free construction, though, cannot commit to anything in
advance, and needs a different device: a criterion built from the
actual rate a candidate induces, not from any type it belongs to
(Section~\ref{sec:typefree}).

The second, deeper and less obvious, is on the converse side. At a
single stage, there is one achievable rate and one converse binding
every code at once; at two stages, the region is a union over many
pieces, one per type of joint behavior --- already Rimoldi's own
structure. There, concentration ties a successful code to a single
auxiliary law, so a converse proved at that law already binds it;
with an individual sequence and an arbitrary code, no such
concentration exists, so a converse proved at just one type leaves a
real gap. Section~\ref{sec:types} closes it by checking every type
at once --- needed only once type-freedom is on the table, since a
code committed to one type (Theorems~\ref{thm:ach} and
\ref{thm:final}) is already bound by that single check.

\subsection{Related work}
\label{sec:related}

This work lies at the intersection of three lines of research, and
draws throughout on the classical theory of rate-distortion coding
\cite{berger71,coverthomas06,gallager68}.

\emph{Universal lossy compression (probabilistic sources).} A long
line of work bounds the redundancy achievable when compressing a
random source of unknown statistics, from the $\Theta(\log
n/n)$-order results of \cite{zhangyangwei97,yuspeed93} and the
ergodic-source treatment of \cite{ornsteinshields90}, to
\cite{kontoyiannis00}'s CLT for the redundancy and, with
\cite{kontoyianniszhang02}, a characterization of optimal compression
via the negative log-probability of a distortion sphere --- the
probabilistic-setting analogue of this paper's own sphere-probability
quantity, and the one connection from this line drawn on directly
below. Mahmood and Wagner \cite{mahmoodwagner22a,mahmoodwagner22b} treat codes universal
in both source and distortion measure.

\emph{The individual-sequence approach.} Pioneered by Ziv, assuming
no source statistics at all and imposing finite-state limitations on
the encoder/decoder instead \cite{ziv78,ziv80,zivlempel78}, later
extended to side information \cite{ziv84}. Article
\cite{merhavcohen20} gives the implementation this paper relies on
(Section~\ref{sec:lz}): realizing a random-coding ensemble via an LZ
decompressor fed by random bits. Article \cite{merhavguessing20} studies the
same kind of device for a different purpose --- optimally guessing a
secret sequence via random bits fed into an LZ78 decoder --- and
shows the guessing effort needed is governed by the 
finite-state compressibility of the given sequence. 
Article \cite{merhav233}, which serves as the single-stage basis for this
paper, combines this philosophy with
\cite{kontoyianniszhang02}'s characterization, proposing
the random coding distribution $U(\bxh)\propto2^{-LZ(\bxh)}$ specifically, with a matching converse
(for the vast majority of codewords within a type) and pointwise
achievability. 

\emph{Successive refinement.} Equitz and Cover \cite{equitzcover91}
proved, for memoryless sources, that refining $D_1$ into
$D_2\le D_1$ achieves both rate-distortion optima simultaneously iff
the optimal reproductions satisfy $X-\hat X_2-\hat X_1$ --- necessary
\emph{and} sufficient for the specific pair, via a chain rule built
on El~Gamal and Cover's \cite{elgamalcover82} own multiple-description
achievable region; Koshelev \cite{koshelev94} and Kanlis and Narayan
\cite{kanlisnarayan96} obtained related divisibility and
error-exponent results. Rimoldi \cite{rimoldi94}
characterized the full region as a union over auxiliary
distributions, the template Section~\ref{sec:types} follows.
\cite{merhav244}, this paper's direct predecessor, is discussed in
detail in Section~\ref{sec:relation244}.

\subsection{Why Lempel--Ziv complexity}
\label{sec:motivation}

As mentioned earlier, article \cite{merhav233} already establishes, for the single-stage problem,
that an LZ78-based random-coding ensemble has several properties no
type-restricted construction can match. Section~\ref{sec:types}
builds a purely type-based achievability/converse pair for the
two-stage problem --- a direct translation of \cite{rimoldi94}'s own
classical achievable-region characterization to individual sequences,
included only to supply a matching converse, not as a contribution in
its own right; this paper calls it the type-covering skeleton throughout.
This paper's question is whether an LZ78-based construction can match
that skeleton's rates exactly while also carrying the same several
properties into the two-stage setting.

\begin{enumerate}
\item \emph{Universality across the entire problem specification.} A
scheme built on the random-coding distribution 
$U_{Q^\ell}$, which is the uniform distribution across a certain type
class $Q^\ell$ --- an empirical distribution over blocks
of length $\ell$
--- requires commitment to a fixed $Q^\ell$ \emph{before} the codebook is built, and
$Q^\ell$ in turn must depend on everything the problem specifies: the
type of the source sequence, both distortion measures, and
both budgets. Also,
Stage~2 needs its own, analogous conditional-type object.
All this is completely avoided by the use of the universal LZ random-coding
distribution and its conditional counterpart for the second stage, which are
universally asymptotically optimal across the type of the source sequences,
the two distortion measures and the two distortion levels.
\item \emph{No commitment to a block length in the codebook
itself.} The choice of $U_{Q^\ell}$ requires also commitment to a certain
choice of $\ell$, and so, changing $\ell$
amounts to regenerating the codebook. By contrast, 
the LZ universal measure references no $\ell$ at all, so
the same codebook serves every choice, and adjusting $\ell$ costs
nothing beyond changing what is searched for.
\item \emph{Actual implementability.} The type-class uniform distribution $U_{Q^\ell}$ is a proof
device: covering-codebook existence is established probabilistically,
with no suggestion of how to efficiently draw from it. On the other hand, 
the LZ78 algorithm is real
and widely deployed; its ensemble is realized concretely by feeding
shared random bits into an LZ78 decompressor \cite{merhavcohen20}
(Section~\ref{sec:lz}'s explicit protocol).
\end{enumerate}

Section~\ref{sec:lz} shows all three properties survive the extension
to two stages, with the same rate Section~\ref{sec:types}'s skeleton
achieves (Theorem~\ref{thm:ach}) matched exactly --- none of them
cost a weaker guarantee. This buys one further thing beyond the three
items above. Section~\ref{sec:types}'s own construction needs a
different codebook for every point of Section~\ref{sec:ach-region}'s
frontier; Section~\ref{sec:typefree}'s type-free construction, built
from the same pair $(U,V(\cdot|\bxh_1))$ but referencing no target
type at all, lets a single, fixed ensemble realize every point of
the achievable region simultaneously, with which point obtained
selected purely by the search --- though this search criterion pays
a real price for the convenience, checked against the full remaining
problem for each candidate
tried, unlike the type-committed route's simpler, single-candidate check
(Section~\ref{sec:tf-converse}'s own closing discussion).

None of this is automatic. Progressive decodability --- the first
$L_1(\bx)$ bits alone (Stage~1's own codelength) must already give a
valid $D_1$-reconstruction, before the next $L_2(\bx)$ bits
(Stage~2's incremental codelength) exist --- means $\bxh_1$ must also
serve as good side information for Stage~2, not merely satisfy $D_1$
on its own. Section~\ref{sec:types} explains why, in a classical
argument tracing back to Rimoldi's own characterization;
Section~\ref{sec:why-hard} confirms the same requirement survives,
unchanged, once Stage~1 draws from the unrestricted $U$ instead.

\subsection{Relation to the finite-state-encoder approach of
\cite{merhav244}}
\label{sec:relation244}

Article \cite{merhav244} solves the identical two-stage problem, but differs
from this paper in both modeling and achievability. Its encoder model
comprises a general vector quantizer (VQ) followed by a finite-state
lossless encoder operating on the reconstruction, with convergence
only after a double limit ($n\to\infty$, then $q\to\infty$, $q$ being the
number of states); once that limit is taken, however, the resulting
guarantee holds for each individual sequence, not merely on average
over some ensemble.
This paper's encoder, on the other hand, carries no finite-state restriction at all.
This trade is already known to be favorable at the
single-stage level: Article \cite{merhav233} already shows, even
there, that the unrestricted route is never worse --- and in
general, strictly better --- than the VQ-plus-finite-state-encoder
paradigm, which in turn yields the minimum LZ compression ratio across the $D$-ball
around the given source sequence.
The intuitive reason is that the finite-state compressor and the
LZ78 compressor alike act on the VQ's own output, yet neither is
built to exploit the fact that this output --- the actual input to
the lossless encoder --- is confined to a sparse subset of
$\Xh^n$, since not every sequence there can arise as VQ output. This
sparsity goes unexploited, creating considerable room for
improvement, exploited in \cite{merhav233}.
Likewise, the achievable rate region of 
\cite{merhav244} is subsumed by that of the present paper, as
the converse of \cite{merhav244} is also framed in the encoder model of a VQ
followed by $q$-state encoders. 

\subsection{Outline}

Section~\ref{sec:setup} establishes the notation and the definitions needed.
Section~\ref{sec:types} builds a complete direct/converse
match in the language of types.
Section~\ref{sec:lz} provides the LZ-based random coding scheme.
Section~\ref{sec:typefree} then gives a fully type-free
construction.
Section~\ref{sec:tf-converse} shows it matches
a converse exactly.
Section~\ref{sec:domination}
then uses this same type-free construction to show this paper's
achievable region is a superset of \cite{merhav244}.
Section~\ref{sec:mstage} extends the type-free construction to a
general fixed number of stages. 
Finally, Section~\ref{sec:discussion} summarizes.

\section{Formulation and Notation Conventions}
\label{sec:setup}

Let $n$ be a given positive integer and $\ell$
divide $n$. Define $m=\frac{n}{\ell}$. Let $\bz\in\mathcal A^n$
denote a generic sequence over an alphabet $\mathcal A$.
Partition $\bz$ into $m$
non-overlapping $\ell$-blocks $z_1,\dots,z_m\in\mathcal A^\ell$. Its
\emph{block type} is defined as the empirical distribution of these $m$ blocks
over $\mathcal A^\ell$, denoted $\hat P_{\bz}$, i.e.,
\begin{equation}
\hat P_{\bz}(a)\dfn\frac1m\big|\{i\in\{1,\dots,m\}:z_i=a\}\big|,~~~
a\in\mathcal A^\ell. 
\end{equation}
A \emph{block permutation} $\pi$ acts on $\bz$
by permuting its $m$ blocks as whole units; obviously, block permutations
preserve the block type, i.e.\
$\hat P_{\pi(\bz)}=\hat P_{\bz}$.
Likewise, for a fixed number $k$ of
length-$n$ sequences $\bz^{(1)},\dots,\bz^{(k)}$, over alphabets
$\mathcal A_1,\dots,\mathcal A_k$ respectively,
write $\bz^{(1)}\odot\cdots\odot\bz^{(k)}$ for their
\emph{product sequence}, namely, the length-$n$ sequence over the product
alphabet $\mathcal A_1\times\cdots\times\mathcal A_k$ whose $i$-th
symbol is $(z^{(1)}_i,\dots,z^{(k)}_i)$. Their \emph{joint (block)
type} is denoted
\begin{equation}
\hat P_{\bz^{(1)}\dots\bz^{(k)}}\dfn\hat
P_{\bz^{(1)}\odot\cdots\odot\bz^{(k)}}.
\end{equation}
Clearly, a common block permutation
applied to all $\bz^{(i)}$, $i=1,\dots k$, simultaneously, preserves
this joint type.
Denote the fixed source sequence by
$\bx$, with block type $P^\ell\dfn\hat P_{\bx}$ and finite alphabet 
$\mathcal X$ of cardinality $K=|\mathcal X|$. The reproduction alphabets of the
two reconstruction stages are denoted
by $\Xh_1$ and $\Xh_2$
with $\hat K_1=|\Xh_1|$, $\hat K_2=|\Xh_2|$; where a single, generic $\Xh$ or
$\hat K$ appears below (Lemma~\ref{lem:dc}, stated once for either
stage), it stands for whichever of $\Xh_1,\Xh_2$ (and $\hat K_1, \hat K_2$) is
being instantiated, exactly as with $d,D$ below.

Let $\bxh_1\in\Xh_1^n$ and
$\bxh_2\in\Xh_2^n$ denote the reconstruction vectors of the first stage and
second stage, respectively.
The distortion functions of the first and the second stage, $d_1$ and $d_2$,
both depend on $\bx$ and the corresponding reconstruction ($\bxh_1$ and
$\bxh_2$, respectively) only via their
first-order (single-symbol) joint empirical distribution, and, for
fixed such distribution, they grow linearly in $n$, which is naturally the
case with additive distortion functions.

Define also $\mathcal S(\bx,D_1)=\{\bxh_1:d_1(\bx,\bxh_1)\le nD_1\}$, 
$\hat{\mathcal S}(\bxh_1,D_1)=\{\bx:d_1(\bx,\bxh_1)\le nD_1\}$, and analogously
$\mathcal S(\bx,D_2),\hat{\mathcal S}(\bxh_2,D_2)$ with $d_2$ in place of $d_1$;
$\mathcal S_2(\bx,D_1,D_2)=\mathcal S(\bx,D_1)\times \mathcal S(\bx,D_2)$,
$\hat{\mathcal S}_2(\bxh_1,\bxh_2,D_1,D_2)=\{\bx:(\bxh_1,\bxh_2)\in
\mathcal S_2(\bx,D_1,D_2)\}$. For a joint type $F$, define
$\mathcal T(F|\bx)\dfn\{\bxh:~\hat P_{\bx\bxh}=F\}$, 
along with an extension in the same way to more than one conditioning sequence,
$\mathcal T(F|\bx,\bxh_1)\dfn\{\bxh_2:~\hat P_{\bx\bxh_1\bxh_2}=F\}$
and so on. Wherever a single, generic $d$ or $D$ appears below
(Lemma~\ref{lem:dc}, stated once for either stage), it stands for
whichever of $d_1,d_2$ (and $D_1,D_2$) is being instantiated.

A \emph{two-stage code} $\Phi=(\phi_1,\phi_2,\psi_1,\psi_2)$ consists
of two encoders $\phi_1,\phi_2:\mathcal X^n\to\{0,1\}^*$, with images
$\calB_1\dfn\phi_1(\mathcal X^n)$ and $\calB_2\dfn\phi_2(\mathcal
X^n)$ --- in general distinct subsets of $\{0,1\}^*$, the set of
finite variable-length binary strings --- and two decoders
$\psi_1:\calB_1\to\Xh_1^n$, $\psi_2:\calB_1\times\calB_2\to
\Xh_2^n$; the encoder output on $\bx$ is
$(\bxh_1,\bxh_2)\dfn(\psi_1(\phi_1(\bx)),\psi_2(\phi_1(\bx),\phi_2(\bx)))$,
with codeword lengths $L_1(\bx)\dfn|\phi_1(\bx)|$,
$L_2(\bx)\dfn|\phi_2(\bx)|$ and data-dependent rates
$\rho_1(\bx)\dfn L_1(\bx)/n$, $\rho_2(\bx)\dfn L_2(\bx)/n$.
The notations $R_1$ and $R_2$ are reserved for
the coordinates of the rate-pair plane.
The code $\Phi$ \emph{meets both distortion constraints}
if $d_1(\bx,\bxh_1)\le nD_1$ and $d_2(\bx,\bxh_2)\le nD_2$ for every
$\bx$; it gives a \emph{one-to-one correspondence between codewords
and binary strings at each stage} if $\psi_1$ is injective on
$\calB_1$ and $(u,v)\mapsto(\psi_1(u),\psi_2(u,v))$ is
injective on $\{(\phi_1(\bx),\phi_2(\bx)):\bx\in\mathcal X^n\}$ ---
distinct reconstructions get distinct binary representations, at each stage and
jointly, with neither prefix-freedom nor unique decodability
required (all Theorem~\ref{thm:conv}'s own counting argument needs).
Both conditions are assumed throughout.

For a type $Q^\ell$ on $\Xh^\ell$, $\mathcal T_n(Q^\ell)\dfn\{\bxh\in\Xh^n:
~\hat P_{\bxh}=Q^\ell\}\subset\Xh^n$ is its
type class. Let $U_{Q^\ell}(\cdot)$ denote the uniform distribution on it,
\begin{equation}
U_{Q^\ell}(\bxh)=\left\{\begin{array}{ll}
\frac{1}{|\mathcal T_n(Q^\ell)|} & \bxh\in\mathcal T_n(Q^\ell)\\
0 & \mbox{elsewhere,}\end{array}\right.
\end{equation}

\emph{Notation of marginals and conditionals of $W^\ell$.} 
Let $X$, $\hat X_1$ and $\hat X_2$ denote auxiliary random variables, taking
values on $\calX^\ell$, $\hat{\calX}_1^\ell$, and $\hat{\calX}_2^\ell$,
respectively, jointly distributed according to a given joint type
of $\ell$-vectors $\hat{P}_{\bx\bxh_1\bxh_2}=W$.
For any
subset of $\{X,\hat X_1,\hat X_2\}$, possibly with a further subset
as conditioning, $W^\ell$ subscripted by that subset --- with a bar
for conditioning, following the customary notation rules of probability
theory, denotes the corresponding marginal and/or conditional type induced by $W^\ell$,
e.g.\ $W_X^\ell$ is $W^\ell$'s own $X$-marginal, $W_{X\hat X_1}^\ell$
its own $(X,\hat X_1)$ sub-type, $W_{\hat X_2|\hat X_1}^\ell$ its
own conditional type of $\hat X_2$ given $\hat X_1$, etc.
For the sake of notational simplicity, the few such
objects used repeatedly throughout also carry a short
alias, defined here once and used interchangeably with the
subscripted form thereafter:
\begin{equation}
P^\ell\dfn W_X^\ell,\qquad Q_1^\ell\dfn W_{\hat X_1}^\ell,\qquad
Q_2^\ell\dfn W_{\hat X_2}^\ell,\qquad
W_1^\ell\dfn W_{X,\hat X_1}^\ell,\qquad W_2^\ell\dfn W_{X,\hat X_2}^\ell.
\end{equation}
Objects used only rarely
are left in the fully
subscripted form.

For the type $Q_1^\ell$, treated as a probability distribution on its
own alphabet, $\hat H_{Q_1^\ell}(\hat X_1)\dfn-\sum_a Q_1^\ell(a)\log Q_1^\ell(a)$ is
its (empirical) entropy; for a joint type $W_1^\ell$, $\hat
H_{W_1^\ell}(\hat X_1|X)$, $\hat I_{W_1^\ell}(X;\hat X_1)=
\hat H_{Q_1^\ell}(\hat X_1)-\hat
H_{W_1^\ell}(\hat X_1|X)$
denote the corresponding conditional entropy and mutual
information, and so on.
We will also use alternative notations like $\hat H(Q_1^\ell)$ and 
$\hat H(Q_2^\ell)$ for 
$\hat H_{Q_1^\ell}(\hat X_1)$
and $\hat H_{Q_2^\ell}(\hat X_2)$ whenever
convenient and safe from any risk of ambiguity.

\subsection{A general random-coding tool}
\label{sec:universal-ach}

One achievability mechanism is used throughout this
paper: search a large randomly and independently selected codebook, drawn from some fixed
random coding distribution, for the first candidate meeting a certain target set --- and the
resulting codelength is governed entirely by the target's 
probability under that measure. Both pieces vary from one use to the
next, but the underlying principle remains the same.
The following lemma is proved in Appendix~\ref{app:proofs} very similarly to the proof of
Theorem~2 of \cite{merhav233}.

\begin{lemma}[Universal random-coding achievability]
\label{lem:universal-ach}
Let $\mu$ be a probability distribution on $\Xh^n$ with
$\mu(\bxh)\ge B^{-n}$ for every $\bxh$ in its support, for some
$B\ge1$ (possibly depending on $n$). For every $A>B$ and $\epsilon>0$
there is $N=N(\epsilon,A,B)$ such that, for every $n>N$, there is a
codebook of $A^n$ sequences, drawn independently at random from
$\mu$, such that for every $\bx$ and every target
$\mathcal E(\bx)\subset\Xh^n$, the first codeword landing in
$\mathcal E(\bx)$ gives codelength
\begin{equation}
L(\bx)\le-\log\mu[\mathcal E(\bx)]+(2+\epsilon)\log n+c,
\end{equation}
$c>0$ depends only on $A$. The restriction to $n>N(\epsilon)$ cannot
be dropped: $N(\epsilon)\to\infty$ as $\epsilon\to0^+$, and the
argument fails outright at $\epsilon=0$ --- see the proof.
\end{lemma}

Intuitively, each candidate lands in $\mathcal E(\bx)$ independently
with probability $\mu[\mathcal E(\bx)]$, so the first success typically
takes about $1/\mu[\mathcal E(\bx)]$ draws, and encoding this number
costs roughly about $\log\big(1/\mu[\mathcal E(\bx)]\big) =
-\log\mu[\mathcal E(\bx)]$ bits, which is the bound's leading term.

\section{The type-covering skeleton}
\label{sec:types}

Even in Rimoldi's own, classical (probabilistic, non-universal)
setting --- with $X,\hat X_1,\hat X_2$ single-symbol random
variables whose joint law is the classical auxiliary distribution,
and the actual scheme operating on the length-$n$ sequences
$X^n,\hat X_1^n,\hat X_2^n$, each drawn i.i.d.\ according to it ---
the two-stage rate region is not achieved by letting
Stage~1 choose $\hat X_1^n$ to merely satisfy the first distortion
constraint, then designing an optimal Stage-2 code for whatever
$\hat X_1^n$ results. The standard covering-lemma proof of Stage~1's
own achievability draws candidate sequences i.i.d., each coordinate
from the auxiliary's own
marginal law, and searches for the first one \emph{jointly typical}
with $X^n$ under the specific joint law --- not merely distortion-compliant.
Bounding Stage~2's rate, given only that $\hat X_1^n$ meets $D_1$,
requires an expectation of the form $E\{-\log(\mbox{Stage-2 success
probability}\mid\hat X_1^n)\}$ over that uncontrolled, random $\hat
X_1^n$; since $-\log(\cdot)$ is convex, Jensen's inequality bounds this
only from below, the wrong direction for achievability. The gap is
between two same-size Stage-2 codebooks: one with codewords drawn
fully independently from the compound law, achieving Jensen's own
lower bound; the real scheme's own, whose codewords form a single
\emph{cloud} around whichever $\hat X_1^n$ results, conditionally
--- not marginally --- independent, achieving the true, larger rate
instead. This is
exactly why Rimoldi's characterization is stated as a union over
auxiliary \emph{joint} distributions of $(X,\hat X_1,\hat X_2)$, not
over distortion-feasible marginals: fixing $\hat X_1$'s relationship
to $X$ at the single-letter level pins down $\hat X_1^n$'s
relationship to $X^n$ throughout, removing the expectation, hence
the Jensen gap, entirely. This
paper's own joint-type matching, throughout what follows, is the
individual-sequence transplant of that same cure.

\subsection{A single double-counting lemma, instantiated twice}

The following lemma and its proof already appear in \cite[Theorem 1]{merhav233}.
We restate it here, in the
single-reconstruction case first, for the sake of completeness, and then
note separately how it extends to the 
two-stage reconstruction case.

\begin{lemma}[Double counting] \label{lem:dc}
Given two types $P^\ell$ (on $\mathcal X^\ell$) and $Q^\ell$ (on
$\Xh^\ell$) the following holds. For any two representatives, $\bx\in \mathcal T_n(P^\ell)$ and 
$\bxh\in \mathcal T_n(Q^\ell)$,
\begin{equation}
\frac{|\mathcal T_n(Q^\ell)\cap\mathcal S(\bx,D)|}{|\mathcal T_n(Q^\ell)|}
= \frac{|\mathcal T_n(P^\ell)\cap\hat{\mathcal S}(\bxh,D)|}{|\mathcal T_n(P^\ell)|},
\end{equation}
and both sides are independent of the particular representatives chosen.
\end{lemma}

The significance of this lemma is as follows. The left-hand side is readily
interpreted as $U_{Q^{\ell}}[\calS(\bx,D)]$, which is the probability that a
randomly selected codeword under $U_{Q^\ell}$ would fall within distance $nD$
from $\bx$. The right-hand side is the reciprocal of the 
sphere-covering ratio in the input space. While the former is intimately related to random-coding
achievability, the latter is a central building block of the converse.
Intuitively, Lemma \ref{lem:dc} can be considered as a
combinatorial analogue (considering the method of types) to the simple fact that the rate-distortion function can
be either represented as the minimum of $H(X)-H(X|\hat{X})=I(X;\hat{X})$ or as the minimum
of $H(\hat{X})-H(\hat{X}|X)=I(X;\hat{X})$, both taken subject to the distortion constraint.

\begin{proof}[Proof sketch]
Let $N(D)\dfn\sum_{\bx,\bxh}\bone\{\bx\in \mathcal T_n(P^\ell),~\bxh\in
\mathcal T_n(Q^\ell),~d(\bx,\bxh)\le nD\}$. Counting it by grouping over
$\bx$ first, or over $\bxh$ first, gives
$N(D)=|\mathcal T_n(P^\ell)|\cdot|\mathcal T_n(Q^\ell)\cap \mathcal S(\bx,D)|$ (any
representative $\bx$) and $N(D)=|\mathcal T_n(Q^\ell)|\cdot|\mathcal 
T_n(P^\ell)\cap\hat{\mathcal S}(\bxh,D)|$ (any representative $\bxh$): each equality uses that
distortion depends only on the first-order joint empirical
distribution (Section~\ref{sec:setup}), so any blockwise permutation
carrying one representative to another carries the corresponding
intersection onto itself bijectively, leaving its size unchanged ---
hence $|\mathcal T_n(Q^\ell)\cap \mathcal S(\bx,D)|$ does not depend on which $\bx$ is
chosen, nor $|\mathcal T_n(P^\ell)\cap\hat{\mathcal S}(\bxh,D)|$ on which $\bxh$ is
chosen. Equating the two expressions for $N(D)$ and dividing by
$|\mathcal T_n(P^\ell)|\cdot|\mathcal T_n(Q^\ell)|$ gives the claim. 
\end{proof}

\emph{Extension to two reconstructions, and to type-match targets.}
This extends Lemma~\ref{lem:dc} to the pair $(\bxh_1,\bxh_2)$,
treated as a single reconstruction over the
product alphabet $\Xh_1\times\Xh_2$, exactly as a single $\bxh$ was
treated over $\Xh$ above. The lemma applies verbatim:
$\mathcal S_2(\bx,D_1,D_2)$, and likewise $\mathcal T(W^\ell|\bx)\dfn\{\bxh:\hat P_{\bx,\bxh}=W^\ell\}$
--- now with $\bxh=(\bxh_1,\bxh_2)$ ranging over the product alphabet
--- are each
still invariant under a common blockwise permutation of $\bx$ and the
reconstruction, for exactly the same reason $\mathcal S(\bx,D)$ is
(Section~\ref{sec:setup}) --- two distortion constraints, or a full
type match, are no different from one in this respect. So nothing
changes beyond relabeling: $U_{W^\ell}[\mathcal S_2(\bx,D_1,D_2)]=|\mathcal T_n(P^\ell)
\cap\hat S_2|/|\mathcal T_n(P^\ell)|$ and $U_{Q^\ell}[\mathcal
T(W^\ell|\bx)]=|\mathcal T_n(P^\ell)\cap\mathcal T(W^\ell|\bxh)|/|\mathcal T_n(P^\ell)|$
--- the form the rest of this paper uses throughout, for both one
and two reconstructions. This concludes the extension.

The single-reconstruction case --- target $\mathcal T(W_1^\ell|\bx)$
over the alphabet $\Xh_1^\ell$ --- is what Section~\ref{sec:types}
uses throughout; the two-stage case --- target $\mathcal T(W^\ell|\bx)$
(the full-triple type match) over the doubled alphabet
$\Xh_1^\ell\times\Xh_2^\ell$ --- is what the proofs below use. Both
are used without re-deriving the lemma.

\subsection{Converse}

Lemma~\ref{lem:dc} gives the bound this subsection is built from: a
rate lower bound for any code, holding for the vast majority of
source sequences of a given type, but conditional on the code
actually realizing that type on the sequence in question.
Theorem~\ref{thm:conv} below states this bound formally, its proof
deferred to Appendix~\ref{app:thm1}; it serves
as a preparatory, supporting step toward the actual converse,
Theorem~\ref{thm:conv-region}, established next.

\begin{theorem}[Type-conditional rate bound] \label{thm:conv}
Let $\Phi=(\phi_1,\phi_2,\psi_1,\psi_2)$ be any two-stage code that
satisfies the two distortion constraints. Fix any joint type
$W^\ell$ with $X$-marginal $W_X=P^\ell$ and let $\epsilon$ be an
arbitrarily small positive real. For at
least a $(1-2n^{-\epsilon})$ fraction of the source sequences $\{\bx\}$ in $\mathcal
T_n(P^\ell)$ the following holds: If
$(\bx,\psi_1(\phi_1(\bx)),\psi_2(\phi_1(\bx),\phi_2(\bx)))$ falls in
joint type $W^\ell$, then
\begin{equation}
L_1(\bx) \ge nR_1^*(W^\ell) - \epsilon\log n, \qquad
L_1(\bx)+L_2(\bx) \ge nR^*(W^\ell) - \epsilon\log n,
\end{equation}
where
\begin{equation}
R_1^*(W^\ell)\dfn-\frac1n\log\big(U_{Q_1^\ell}[\mathcal T(W_1^\ell|\bx)]\big),
\qquad
R^*(W^\ell)\dfn-\frac1n\log\big(U_{W^\ell}[\mathcal T(W^\ell|\bx)]\big).
\end{equation}
\end{theorem}

The single-description bound is the special case of the marginal
bound with the second stage absent.

\emph{Why this is not yet a converse.} Theorem~\ref{thm:conv}'s
guarantee is conditional on the event described in the ``if'' clause
above --- but $\Phi$ and $W^\ell$ are chosen independently of each
other, and nothing forces $\Phi$ to realize that particular type on
more than a negligible subset of $\mathcal T_n(P^\ell)$. Concretely:
for a fixed $\Phi$, let
$G_\Phi(W^\ell)\dfn\{\bx\in\mathcal T_n(P^\ell):
(\bx,\psi_1(\phi_1(\bx)),\psi_2(\phi_1(\bx),\phi_2(\bx)))\mbox{ has
joint type }W^\ell\}$. Theorem~\ref{thm:conv} only
bounds how much of $G_\Phi(W^\ell)$ can violate the rate bound, as a
fraction of $|\mathcal T_n(P^\ell)|$ --- not as a fraction of
$|G_\Phi(W^\ell)|$ itself. A code that realizes type $W^\ell$ only on
a small, handpicked subset of $\mathcal T_n(P^\ell)$ (spreading its
remaining output thinly across other types) can make $G_\Phi(W^\ell)$
itself smaller than Theorem~\ref{thm:conv}'s own exceptional set, and
so can violate the rate bound on \emph{all} of $G_\Phi(W^\ell)$
without contradicting the theorem as stated. Turning
Theorem~\ref{thm:conv} into an actual converse takes tying the
exceptional-set size to the number of types available, not to a
single, pre-chosen $W^\ell$ in isolation --- done in
Theorem~\ref{thm:conv-region} next, by union-bounding over every type in $\mathcal W$
simultaneously.

Let $\mathcal W$ be the family of types with $X$-marginal $P^\ell$
meeting both distortion budgets, with $\mathcal T_n(W^\ell)\ne\emptyset$;
since $\ell=\ell(n)=\Theta(\log n)$ is fixed at each $n$ and the
alphabets are fixed, $|\mathcal W|\le n^c$ for a constant
$c=c(|\mathcal X|,\hat K_1,\hat K_2)$. Write
\begin{equation}
\mathcal R(\bx) \ \dfn\ \bigcup_{W^\ell\in\mathcal W} \big\{(R_1,R_2): R_1\ge R_1^*(W^\ell),\ R_1+R_2\ge R^*(W^\ell)\big\}
\end{equation}
for the union of every type's own quadrant ---
Section~\ref{sec:ach-region} below shows this is exactly the
achievable region, once Theorem~\ref{thm:ach}'s own achievability is
in place.

\begin{theorem}[Type converse] \label{thm:conv-region}
For every type $P^\ell$ (on $\mathcal X^\ell$), every two-stage code
$\Phi$ that satisfies both distortion
constraints, and every $\epsilon>0$, there is at most a $2n^{-\epsilon}$
fraction of $\bx\in \mathcal T_n(P^\ell)$ for which $\Phi$'s own rate
pair $\big(\rho_1(\bx),\rho_1(\bx)+\rho_2(\bx)\big)$ falls outside
$\mathcal R(\bx)$, up to $O(\log n/n)$ slack in each coordinate.
\end{theorem}
\begin{proof}
\emph{Step 1: a single exceptional set, uniform over types.} For
each $W^\ell\in\mathcal W$, apply Theorem~\ref{thm:conv} to $\Phi$ at
this $W^\ell$, with $\epsilon$ there set to $\epsilon+c$: this
excludes a set $E_{W^\ell}\subseteq\mathcal T_n(P^\ell)$ of size at
most $2n^{-\epsilon-c}|\mathcal T_n(P^\ell)|$, outside of which
$\Phi$'s output having type $W^\ell$ implies both
\[
L_1(\bx)\ge nR_1^*(W^\ell)-(\epsilon+c)\log n
\qquad\mbox{and}\qquad
L_1(\bx)+L_2(\bx)\ge nR^*(W^\ell)-(\epsilon+c)\log n
\]
--- a single exceptional set for both bounds, since
Theorem~\ref{thm:conv} already gives them together, not as two
separately-excluded events. Let $E\dfn\bigcup_{W^\ell\in\mathcal
W}E_{W^\ell}$; since $|\mathcal W|\le n^c$, the union bound gives
\[
|E|\ \le\ n^c\cdot2n^{-\epsilon-c}|\mathcal T_n(P^\ell)|\ =\
2n^{-\epsilon}|\mathcal T_n(P^\ell)|.
\]

\emph{Step 2: apply it at whichever type is actually realized.} Fix
any $\bx\notin E$. Since $\Phi$ meets both distortion constraints,
its output on $\bx$ has some joint type $W^{\ell,*}\in\mathcal W$
(distortion depends on $\bx$ and the reconstruction only through
their joint type). As $\bx\notin E_{W^{\ell,*}}$, the bounds above
hold at $W^\ell=W^{\ell,*}$, so
$\big(\rho_1(\bx),\rho_1(\bx)+\rho_2(\bx)\big)$ lies, up to $O(\log
n/n)$ slack, in the quadrant
\[
\{(R_1,R_2):R_1\ge R_1^*(W^{\ell,*}),\ R_1+R_2\ge R^*(W^{\ell,*})\},
\]
one of the quadrants unioned to form $\mathcal R(\bx)$. Hence the
rate pair lies in $\mathcal R(\bx)$, up to the same slack, for every
$\bx\notin E$ --- and $|E|\le2n^{-\epsilon}|\mathcal T_n(P^\ell)|$, as
claimed.
\end{proof}

This holds for every code --- type-committed, type-free, or
otherwise: each such code has its own exceptional set of at most
$2n^{-\epsilon}$ of the $\bx$'s, outside which it cannot achieve a
rate pair outside $\mathcal R(\bx)$, whichever type its own output
happens to realize. It is matched, on the nose, by
both achievability routes ahead, though for different reasons.
Theorem~\ref{thm:ach} (and, via the type-committed LZ route,
Theorem~\ref{thm:final}) realizes a single, chosen type $W^\ell$ on
\emph{every} $\bx$ it is run on, so Theorem~\ref{thm:conv} applies to
it directly, at that one type, with no need for the union-over-types
argument just given. Theorem~\ref{thm:tf-match}'s type-free
construction, by contrast, never commits to a specific type ---
exactly why it needs this result, region-level and unconditional in
$\bx$, rather than Theorem~\ref{thm:conv} applied at a single,
pre-chosen type.

Relative to \cite[Theorem 1]{merhav233}'s own single-stage converse,
this result is the stronger one in the one respect the proof of
Theorem~\ref{thm:conv} addresses: \cite{merhav233} bounds
the fraction of \emph{codewords} (distinct reconstructions)
satisfying the length bound, while the counting argument there, via
Lemma~\ref{lem:dc}'s type-match extension, converts this directly
into a bound on the fraction of \emph{source sequences} instead.

\subsection{Achievability}

The bound just proved is matched exactly: for every joint type
meeting both distortion budgets, a scheme exists reaching that
type's own rate pair. 

\begin{theorem}[Type-class achievability] \label{thm:ach}
Fix any joint type $W^\ell$ with $X$-marginal $P^\ell$, meeting both
distortion constraints, with $\mathcal T_n(W^\ell)\ne\emptyset$. For every
$\epsilon>0$ and all sufficiently large $n$, there is a two-stage scheme achieving:
\begin{equation}
L_1(\bx) \le nR_1^*(W^\ell) + (2+\epsilon)\log n + c, \qquad
L_1(\bx)+L_2(\bx) \le nR^*(W^\ell) + (4+\epsilon)\log n + c,
\end{equation}
with $R_1^*(W^\ell)$ and $R^*(W^\ell)$ as defined in
Theorem~\ref{thm:conv}, and $c=c(\epsilon)$.
\end{theorem}

Dividing by $n$ gives $\rho_1(\bx)\to R_1^*(W^\ell)$,
$\rho_1(\bx)+\rho_2(\bx)\to R^*(W^\ell)$, matching
Theorem~\ref{thm:conv} exactly, for \emph{every} such $W^\ell$: no
further restriction on $W^\ell$ is needed, since achievability and
converse are stated in terms of the same type-probability quantity
throughout.
\begin{proof}
\emph{Stage 1.} Since $\mathcal T_n(W^\ell)\ne\emptyset$, some $\hat{\bx}_1$
achieves $\hat P_{\bx\hat{\bx}_1}=W_1^\ell$ for $\bx\in
\mathcal T_n(P^\ell)$, so $u_1\dfn U_{Q_1^\ell}[\mathcal T(W_1^\ell|\bx)]>0$
and $R_1^*(W^\ell)<\infty$. $Q_1^\ell$, being uniform on a type class
of size at most $\hat K_1^n$, satisfies $Q_1^\ell(\bxh)=1/|\mathcal T_n(Q_1^\ell)|\ge \hat K_1^{-n}$ for every $\bxh$ in its
support: Lemma~\ref{lem:universal-ach}'s hypothesis, with $B=\hat K_1$.
Draw independent candidates $\bz$ i.i.d.\ $\sim
Q_1^\ell$ and search for the first one achieving
exactly joint type $W_1^\ell$ with $\bx$; this is
Lemma~\ref{lem:universal-ach} with $\mu=Q_1^\ell$ (any $A>\hat K_1$) and
$\mathcal T(W_1^\ell|\bx)$ in place of the target, giving
$L_1(\bx)\le-\log(u_1)+(2+\epsilon)\log n+c$, once $u_1$ itself is
computed.

Treating the $m=n/\ell$ blocks of $(\bx,\bz)$ as symbols over the
super-alphabets $\mathcal X^\ell,\Xh_1^\ell$: by the
method-of-types, the probability of a random i.i.d.\
sequence falling into a given type decays exponentially in the
divergence between that type and the underlying
distribution (e.g.\ \cite{csiszarkorner11}) --- applied here to the
pair-valued sequence $(\bx_i,\bz_i)$, $i=1,\dots,m$, with $\bx$ fixed
(so its own randomness is degenerate) and $\bz$'s blocks i.i.d.\
$\sim Q_1^\ell$, giving underlying distribution $P^\ell\times
Q_1^\ell$ and target type $W_1^\ell$ ---
\begin{equation}
u_1=\Pr\big[(\bx,\bz)\mbox{ has type }W_1^\ell\big]=2^{-mD(W_1^\ell\|P^\ell\times Q_1^\ell)+O(\log m)}.
\end{equation}
Since $D(W_1^\ell\|P^\ell\times Q_1^\ell)=\hat I_{W_1^\ell}(X;\hat
X_1)$ whenever $W_1^\ell$'s own $X$-marginal is $P^\ell$ (the
standard identity relating a joint type's divergence from the
product of its own marginals to its mutual information), and $\hat
I_{W_1^\ell}(X;\hat X_1)/\ell=R_1^*(W^\ell)$ by definition
(Theorem~\ref{thm:conv}),
\begin{equation}
u_1=\exp_2\{-m\hat I_{W_1^\ell}(X;\hat X_1)+O(\log m)\}=2^{-nR_1^*(W^\ell)+O(\log n)},
\end{equation}
giving $L_1(\bx)\le nR_1^*(W^\ell)+(2+\epsilon)\log n+c$ directly.

\emph{Stage 2.} Given $\hat{\bx}_1$,
draw $\hat{\bx}_2$'s blocks i.i.d., the $i$-th block sampled
$\sim W_{\hat X_2|\hat X_1}^\ell(\cdot|\hat x_{1,i})$ ($W^\ell$'s own $(\hat
X_2|\hat X_1)$ conditional law), and search for the first one
completing the \emph{full} triple to joint type $W^\ell$ with
$(\bx,\hat{\bx}_1)$; write $p$ for this event's per-candidate
probability.

The identical estimate applies conditionally, treating the type $W_1^\ell$ of $(\bx,\hat{\bx}_1)$
as fixed and $\hat{\bx}_2$'s blocks as
drawn i.i.d.\ from the conditional law $W_{\hat X_2|\hat X_1}^\ell(\cdot|a_1)$
given each block's own $\hat{\bx}_1$-value $a_1$: standard conditional
type-counting (again \cite{csiszarkorner11}) gives
\begin{equation}
p=\Pr[\mbox{full type }W^\ell]=\exp_2\{-mD(W^\ell\|W_1^\ell\times W_{\hat X_2|\hat
X_1}^\ell|W_1^\ell)+O(\log m)\},
\end{equation}
where $D(W^\ell\|W_1^\ell\times W_{\hat X_2|\hat X_1}^\ell|W_1^\ell)\dfn
\sum_{a,a_1}W_1^\ell(a,a_1)D(W^\ell(\cdot|a,a_1)\|W_{\hat X_2|\hat
X_1}^\ell(\cdot|a_1))$
is exactly the conditional mutual information $\hat I_{W^\ell}(X;\hat
X_2|\hat X_1)$ (the same identity as Stage~1's, now applied
conditionally on $\hat X_1$), so $p=\exp_2\{-m\hat I_{W^\ell}(X;\hat
X_2|\hat X_1)+O(\log m)\}$. By the
exact entropy chain rule, $\hat I_{W^\ell}(X;\hat X_1,\hat
X_2)=\hat I_{W_1^\ell}(X;\hat X_1)+\hat I_{W^\ell}(X;\hat X_2|\hat
X_1)$, and the identical type-probability estimate applied to the full
triple directly (uniform sampling of $(\hat X_1,\hat X_2)$ pairs
within their own joint type class, in place of $\hat X_1$ alone)
gives $R^*(W^\ell)=\hat I_{W^\ell}(X;\hat X_1,\hat X_2)/\ell$; so
$p=2^{-n[R^*(W^\ell)-R_1^*(W^\ell)]+O(\log n)}$ exactly, with no
assumption that $X-\hat X_1-\hat X_2$ forms a Markov chain: the
sampling is conditional on $\hat X_1$ alone (legitimate, since $\hat
X_1$ is known to both encoder and decoder after Stage~1), while the
success event checks the full joint type with $(\bx,\hat{\bx}_1)$
together, not merely the $(\hat X_1,\hat X_2)$ marginal relationship.
Applying Lemma~\ref{lem:universal-ach} again, conditionally, gives
$L_2(\bx)\le n[R^*(W^\ell)-R_1^*(W^\ell)]+(2+\epsilon)\log n+c$, so
$L_1(\bx)+L_2(\bx)\le nR^*(W^\ell)+2(2+\epsilon)\log n+2c$, which is the
claimed bound after relabeling $2\epsilon,2c$ as $\epsilon,c$ (i.e.\
for a target $\epsilon$, apply both stages with $\epsilon/2$).
\end{proof}

\subsection{The achievable region and its Pareto frontier}
\label{sec:ach-region}

Ranging the matching converse-achievability pair over every type
gives the full region, not just one corner of it: combining
Theorem~\ref{thm:conv-region}'s own converse with
Theorem~\ref{thm:ach}'s own achievability, $\mathcal R(\bx)$ ---
defined above --- is exactly the achievable region.

The Pareto-optimal points of $\mathcal R(\bx)$ --- those undominated in
both coordinates --- form the lower envelope $R_2^{\min}(R_1) =
\min_{W^\ell\in\mathcal
W:R_1^*(W^\ell)\le R_1}\max(0,R^*(W^\ell)-R_1)$; not every $W^\ell\in
\mathcal W$ need itself be Pareto-optimal, since the union
automatically discards any that are dominated. A specific frontier
point is reached by fixing one coordinate as a budget and minimizing
the other over $\mathcal W$ ($\epsilon$-constraint method): for a
target $R_1$-budget $b$, $\min_{W^\ell:R_1^*(W^\ell)\le b}R^*(W^\ell)-b$;
symmetrically for a total-budget $c$, $\min_{W^\ell:R^*(W^\ell)\le
c}R_1^*(W^\ell)$. Since $\mathcal W$ is a finite set (finitely many
joint types on a finite alphabet, for fixed $n,\ell$), the minimum
defining $R_2^{\min}(R_1)$ is a minimum over finitely many values:
the frontier is always achieved by some $W^\ell\in\mathcal W$, not
merely approached.

Which $W^\ell\in\mathcal W$ to target is not resolved by either
theorem --- it is a free design choice, exactly as the choice of
auxiliary distribution is free in Rimoldi's classical
characterization; the $\epsilon$-constraint recipe above describes
how to make it for a desired operating point.

\section{From types to Lempel--Ziv complexity}
\label{sec:lz}

\subsection{Why this is not automatic}
\label{sec:why-hard}

Section~\ref{sec:motivation} already located the difficulty: $\bxh_1$
must also serve as good side information for Stage~2, not merely
satisfy $D_1$ on its own. Section~\ref{sec:types}'s own joint-type
matching exists precisely to avoid the Jensen obstruction explained
there. The same requirement carries over unchanged when Stage~1
draws from $U$ --- the full, unrestricted, LZ-weighted measure ---
instead of a type-restricted one: nothing about switching measures
removes the need to fix $\hat x_1$'s \emph{type}, not merely its
distortion.

The underlying idea is as follows. Build joint typicality between $\bxh_1$ and $\bx$
directly into Stage~1's search criterion, so Stage~2 conditions on a
candidate is \emph{guaranteed} to keep its
own rate low
(Theorem~\ref{thm:final}). This is the same criterion
Section~\ref{sec:types} already uses for its own, type-restricted
Stage~1; checking it directly against $U$ --- with no shortcut
around it --- is what the rest of this section works out.
For a full description of the random coding scheme, we first pause
to define certain ingredients associated with the LZ78 algorithm
and its conditional version.

For $\bxh_1\in\Xh_1^n$: the \emph{incremental
parsing procedure} of the LZ78 algorithm \cite{zivlempel78}
sequentially parses $\bxh_1$ into phrases, each new phrase being the
shortest string that has not been encountered before as a parsed
phrase, with the possible exception of the last phrase, which may be
incomplete. For example, the incremental parsing of $\hat
x_1=011010011000100$ is $0,1,10,100,11,00,01,00$. Let $c(\bxh_1)$ denote
the resulting number of phrases (eight, in this example), and let
$LZ(\bxh_1)$ denote the length, in bits, of the LZ78 binary compressed
code for $\bxh_1$ --- approximately $c(\bxh_1)\log c(\bxh_1)$, a standard
mechanical fact about the algorithm's own construction. Let
$U(\bxh_1)\propto2^{-LZ(\bxh_1)}$ on $\Xh_1^n$ --- normalizable since
$\Xh_1^n$ is finite, and in fact $\sum_{\bxh_1}2^{-LZ(\bxh_1)}\le1$ by
Kraft's inequality (LZ78 being uniquely decodable), a fact used
later --- the same, fixed,
source-independent measure
throughout; no restriction of its support is used anywhere below.

Similarly, for a pair of sequences ($\bxh_1,\bxh_2)$ of the same length $n$:
apply LZ78's incremental parsing to the sequence of \emph{pairs}
$((\hat x_{2,1},\hat x_{1,1}),\dots,(\hat x_{2,n},\hat x_{1,n}))$,
treating each pair
as one product-alphabet symbol \cite{ziv85}; let $c(\bxh_1,\bxh_2)$
denote the resulting number of distinct phrases, $c'(\bxh_1)$ --- the
number of distinct $\bxh_1$-phrases this same joint parse induces
(generally different from $\bxh_1$'s own, unconditional phrase
count), $\bxh_1(l)$ the $l$-th such phrase, and $c_l(\bxh_2|
\bxh_1)$ the number of its occurrences --- equivalently, the number of
distinct $\bxh_2$-phrases jointly appearing with it --- for
$l=1,\dots,c'(\bxh_1)$. The \emph{conditional LZ complexity} of
$\bxh_2$ given $\bxh_1$ is
\begin{equation}
\rho_{LZ}(\bxh_2|\bxh_1) \;\dfn\; \frac1n\sum_{l=1}^{c'(
\bxh_1)}c_l(\bxh_2|\bxh_1)\log c_l(\bxh_2|\bxh_1)
\end{equation}
\cite{uyematsukuzuoka03} --- the individual-sequence analogue of
conditional entropy, built to exploit empirical correlation between
$\bxh_2$ and $\bxh_1$, not merely $\bxh_2$'s own statistics: $\bxh_1$
enters throughout as side information available to both encoder and
decoder, never as a static resource $\bxh_2$ draws phrases from. Let
$LZ(\bxh_2|\bxh_1)$
denote the codelength of the conditional LZ78 coding scheme built on
this side information \cite{ziv85}: $LZ(\bxh_2|\bxh_1)\le
n\rho_{LZ}(\bxh_2|\bxh_1)+n\hat\epsilon(n)$, $\hat\epsilon(n)\to0$
--- the conditional analogue of the unconditional mechanical bound
above. Let
$V(\bxh_2|\bxh_1)\propto2^{-LZ(\bxh_2|\bxh_1)}$ on $\Xh_2^n$
be the resulting conditional measure --- Stage~2's satellite
codebook is drawn from $V(\cdot|\bxh_1)$ once $\bxh_1$ is fixed.
The following lemma will be useful throughout --- see, e.g., \cite{merhav233},
\cite{merhavguessing20},
as well as
\cite{plotnikweinbergerziv92} and \cite[Lemma 13.5.5]{coverthomas06}
for somewhat different, but intimately related family of inequalities,
collectively referred to as
{\em Ziv's inequality}.

\begin{lemma}[Ziv's inequality] \label{lem:F0}
For any $\bxh\in\calT_n(Q^\ell)$,
\begin{equation}
LZ(\bxh)\le \frac{n\hat H(Q^\ell)}{\ell} + n\Delta_n(\ell), \qquad
\Delta_n(\ell)\to1/\ell.
\end{equation}
\end{lemma}

This is a standard fact about the LZ78 algorithm, well known from
\cite{zivlempel78} (see also \cite{merhav233}), and not proved here.

Consequently, for any $\calS\subseteq \mathcal T_n(Q^\ell)$,
\begin{equation}
\label{eq:zivset}
U(\calS)\ge U_{Q^\ell}(\calS)\cdot2^{-n\Delta_n(\ell)}.
\end{equation}
To see why (\ref{eq:zivset}) holds true, use $U(\bxh)\ge2^{-LZ(\bxh)}$ (which holds thanks to Kraft's inequality, since
the normalizing constant of $U(\cdot)$ is at most 1, owing to the unique
decodability of the LZ78 algorithm), combined with the pointwise
bound $2^{-LZ(\bxh)}\ge2^{-n\hat H(Q^\ell)/\ell-n\Delta_n(\ell)}$
from Lemma~\ref{lem:F0}, and sum over $\bxh\in\calS$; then use the standard
type-size formula
$|\mathcal T_n(Q^\ell)|=2^{n\hat H(Q^\ell)/\ell+O(\log n)}$ (its own $O(\log
n)$ error already folded into $\Delta_n(\ell)$'s own $O(\log n/n)$
term):
\begin{eqnarray}
U(\calS) = \sum_{\bxh\in \calS}U(\bxh) &\ge& \sum_{\bxh\in\calS}2^{-LZ(\bxh)}
\ \ge\ |\calS|\cdot2^{-n\hat H(Q^\ell)/\ell-n\Delta_n(\ell)}\nonumber\\
&=& |\calS|\cdot|\mathcal T_n(Q^\ell)|^{-1}\cdot2^{-n\Delta_n(\ell)}
\ =\ U_{Q^\ell}(\calS)\cdot2^{-n\Delta_n(\ell)},
\end{eqnarray}
using $U_{Q^\ell}(\calS)\dfn|\calS|/|\mathcal T_n(Q^\ell)|$ by definition (no
restriction or renormalization of $U$ needed, since $\calS$ is already
contained in a single type).

This second, set-level bound is the bridge Theorem~\ref{thm:final}
below uses to identify what its search achieves with
Section~\ref{sec:types}'s type-based targets.

\begin{lemma}[Ziv's inequality, conditional version] \label{lem:condF0}
For $\bxh_2$ of the conditional type consistent with $W^\ell$ given
$\bxh_1$:
\begin{equation}
LZ(\bxh_2|\bxh_1)\le n\hat H_{W^\ell}(\hat X_2|
\hat X_1)/\ell+n\Delta_n'(\ell), \qquad \Delta_n'(\ell)\to0.
\end{equation}
\end{lemma}

This is the direct conditional analogue of Lemma~\ref{lem:F0}, for
Ziv's conditional LZ construction, first established in
\cite{merhav00}. 
This bound holds by the following consideration: realize a
conditional block Shannon code on the
$\ell$-blocks of $\hat{\bx}_2$, using the blocks of $\hat{\bx}_1$ as
side information. Such an encoder can be implemented as a finite-state encoder.
Ziv and Lempel's own
finite-state converse, extended to side information available at
both encoder and decoder \cite{ziv85}, lower-bounds any such
encoder's codelength by (essentially) the conditional LZ complexity,
for every individual sequence pair. 

Now, observe that 
by Lemma~\ref{lem:F0}, $LZ(\bxh)\le
n(1+\epsilon_n)\log \hat K$ for every $\bxh\in\Xh^n$ (the worst case,
exactly as in \cite{merhav233}), so $U(\bxh)\ge2^{-LZ(\bxh)}\ge2^{-n(1+\epsilon_n)\log \hat K}$ uniformly: $\mu=U$ satisfies
Lemma~\ref{lem:universal-ach}'s hypothesis with
$B=\hat K^{1+\epsilon_n}\to \hat K$. Lemma~\ref{lem:universal-ach} is used
below, and throughout this section, at $\mu=U$.

For $\bx,\bxh_1$ of joint type $W_1^\ell$, define
$\mathcal T(W^\ell|\bx,\bxh_1)\dfn\{\bxh_2:~
\hat P_{\bx\bxh_1\bxh_2}=W^\ell\}$.
By the 
method of types \cite{csiszarkorner11},
\begin{equation}
\label{eq:bprime}
|\mathcal T(W^\ell|\bx,\bxh_1)| = 2^{m\hat H_{W^\ell}(\hat X_2|X,\hat X_1)+O(\log m)}.
\end{equation}

\begin{lemma}[Stage-2 coverage, for every candidate of the right
type] \label{lem:uniform-cov}
Fix a $D_1$-compatible type $W_1^\ell$.
For every
$\bxh_1\in\mathcal T(W_1^\ell|\bx)$
\begin{equation}
V[\mathcal T(W^\ell|\bx,\bxh_1)|\bxh_1]\ge 2^{-n[R^*(W^\ell)-R_1^*(W^\ell)]-n\Delta_n'(\ell)+o(n)}.
\end{equation}
\end{lemma}
\begin{proof}
By Lemma~\ref{lem:condF0}, every $\bxh_2$ with $(\bxh_1,\bxh_2)$
of type $W_{\hat X_2|\hat X_1}^\ell$
satisfies $LZ(\bxh_2|\bxh_1)\le n\hat
H_{W^\ell}(\hat X_2|\hat X_1)/\ell+n\Delta_n'(\ell)$.
Summing $2^{-LZ(\bxh_2|\bxh_1)}\ge2^{-n\hat
H_{W^\ell}(\hat X_2|\hat X_1)/\ell-n\Delta_n'(\ell)}$ over
$\mathcal T(W^\ell|\bx,\bxh_1)$ and using the count~(\ref{eq:bprime})
$|\mathcal T(W^\ell|\bx,\bxh_1)|=2^{n\hat H_{W^\ell}(\hat X_2|X,\hat X_1)/\ell+o(n)}$
\begin{eqnarray}
V[\mathcal T(W^\ell|\bx,\bxh_1)|\bxh_1] &\ge& |\mathcal
T(W^\ell|\bx,\bxh_1)|\cdot2^{-n\hat H_{W^\ell}(\hat X_2|\hat X_1)/\ell-n\Delta_n'(\ell)}\nonumber\\
&=& 2^{-n[\hat H_{W^\ell}(\hat X_2|\hat X_1)-\hat H_{W^\ell}(\hat X_2| X,\hat X_1)]/\ell-n\Delta_n'(\ell)+o(n)}.
\end{eqnarray}
The bracketed difference is exactly the conditional mutual
information $\hat I_{W^\ell}(X;\hat X_2|\hat X_1)=
R^*(W^\ell)-R_1^*(W^\ell)$ (as established just
above Theorem~\ref{thm:ach}), giving the claim. 
\end{proof}

\begin{theorem}[LZ achievability matches the type-based converse]
\label{thm:final}
Fix any $W^\ell\in\mathcal W$. There is a two-stage LZ78-based
random-coding scheme achieving, for all sufficiently large $n$ and every $\bx\in\mathcal T_n(P^\ell)$,
\begin{eqnarray}
L_1(\bx) &\le& nR_1^*(W^\ell) + n\Delta_n(\ell) + O(\log n),\nonumber\\
L_1(\bx)+L_2(\bx) &\le& nR^*(W^\ell) + n\big[\Delta_n(\ell)+\Delta_n'(\ell)\big] + o(n) + O(\log n).
\end{eqnarray}
\end{theorem}
\begin{proof}
\emph{Stage 1.} Draw $\bxh_1^{(1)},\dots,\bxh_1^{(A^n)}$ i.i.d.\ from
$U$ --- the full, unrestricted, $\bx$-independent measure, $A>\hat K_1$.
Search for the first index with $\bxh_1^{(\cdot)} \in
\mathcal T(W_1^\ell|\bx)$. By
Theorem~\ref{thm:conv}'s own definition of $R_1^*(W^\ell)$,
$U_{Q_1^\ell}[\mathcal T(W_1^\ell|\bx)]=2^{-nR_1^*(W^\ell)}$
directly. By
Lemma~\ref{lem:F0}, $U[\mathcal T(W_1^\ell|\bx)] \ge
U_{Q_1^\ell}[\mathcal T(W_1^\ell|\bx)]\cdot2^{-n\Delta_n(\ell)}
=2^{-nR_1^*(W^\ell)-n\Delta_n(\ell)}$. By
Lemma~\ref{lem:universal-ach} with $T=\mathcal T(W_1^\ell|\bx)$, $L_1(\bx)\le
R_1^*(W^\ell)n+n\Delta_n(\ell)+O(\log n)$.

\emph{Stage 2.} Given this $\bxh_1\in\mathcal T(W_1^\ell|\bx)$ --- guaranteed, not random --- draw a satellite
codebook i.i.d.\ from Ziv's conditional construction
$V(\cdot|\bxh_1)$. Search for the first index
landing in $\mathcal T(W^\ell|\bx,\bxh_1)$ --- a valid target for the
$D_2$ constraint, since $\mathcal T(W^\ell|\bx,\bxh_1)\subseteq \mathcal S(\bx,D_2)$
(the same type-determined-distortion property, applied to $W^\ell$'s
own $(X,\hat X_2)$ sub-marginal). By
Lemma~\ref{lem:uniform-cov} ---
applied here at just this one, Stage-1-selected $\bxh_1$;
Section~\ref{sec:tf-converse} below exploits its full generality ---
$V[\mathcal T(W^\ell|\bx,\bxh_1)|\bxh_1]\ge2^{-n[R^*(W^\ell)-R_1^*(W^\ell)]-n\Delta_n'(\ell)+o(n)}$.
By Lemma~\ref{lem:universal-ach} with $\mathcal E(\bx)=\mathcal T(W^\ell|\bx,\bxh_1)$ and
the conditional measure $V(\cdot|\bxh_1)$ in place of $U$, $L_2(\bx)\le
(R^*(W^\ell)-R_1^*(W^\ell))n+n\Delta_n'(\ell)+o(n)+O(\log n)$.

Combining the two stages gives the claim.
\end{proof}

Dividing by $n$ and applying Proposition~\ref{prop:order} (so
$\Delta_n(\ell),\Delta_n'(\ell)\to0$ as $n\to\infty$, for
$\ell=\ell(n)$ as specified there) gives $\rho_1(\bx)\to
R_1^*(W^\ell)$, $\rho_1(\bx)+\rho_2(\bx)\to R^*(W^\ell)$, matching
Theorem~\ref{thm:conv} exactly at this $W^\ell$ --- for every
$\bx\in\mathcal T_n(P^\ell)$, not merely most, since neither
ingredient the proof combines (Lemma~\ref{lem:universal-ach},
Lemma~\ref{lem:dc}'s exact type symmetry,
Lemma~\ref{lem:uniform-cov}'s coverage of every candidate) admits an
exception.

Searching $\mathcal T(W^\ell|\bx,\bxh_1)$ rather than
$\mathcal S(\bx,D_2)$ directly is a real restriction in general: direct
search achieves at least as good a rate (Section~\ref{sec:motivation}'s
fourth point), and $\mathcal S(\bx,D_2)$, unlike
$\mathcal T(W^\ell|\bx,\bxh_1)$, is not tied to any type, so
Lemma~\ref{lem:condF0} has nothing to say about it and no matching
lower bound is provable this way. But the gap does not actually
manifest here, since Stage~2 is the last stage: waiting for the
first $\bxh_2$ landing directly in $\mathcal S(\bx,D_2)$ is
exponentially equivalent to waiting for the first one jointly
typical, given $\bx,\bxh_1$, with the dominant type within that
ball --- the type that also minimizes the incremental rate
$\hat I_{W^\ell}(X;\hat X_2|\hat X_1)$, manifesting part of the
Pareto-optimization of the rate pair. Direct search is a legitimate,
rate-equivalent alternative here; the type-match route is used
regardless, being the one a matching converse can be proved for.
With a further stage to follow, this equivalence would not hold:
$\bxh_2$ would also need to keep that stage's own rate low, which
the dominant type within $\mathcal S(\bx,D_2)$ need not do --- the
same Jensen obstruction Section~\ref{sec:types} identifies for
Stage~1.

\begin{proposition}[Parameter ordering] \label{prop:order}
Let $K_*=\max(|\mathcal X|,\hat K_1,\hat K_2)$,
$\ell(n)=\lfloor\frac{1}{4}\log n/\log K_*\rfloor$. Then, as
$n\to\infty$, the error terms of Theorems~\ref{thm:conv}--\ref{thm:ach},
Lemma~\ref{lem:F0}, and eq.~(\ref{eq:bprime}) (and their
doubled-alphabet forms) all $\to0$, given
$\Delta_n'(\ell)$ has the same qualitative two-term shape as
$\Delta_n(\ell)$ (Lemma~\ref{lem:condF0}'s status).
\end{proposition}
\begin{proof}
$1/\ell(n)=\Theta(1/\log n)\to0$; $\hat K_1^{\ell(n)},\hat K_2^{\ell(n)},
\hat K_1^{\ell(n)}\hat K_2^{\ell(n)}\le
K_*^{2\ell(n)}=n^{1/2+o(1)}$, so each alphabet-size error term is
$O(n^{-1/2+o(1)}\log n)\to0$.
\end{proof}

\subsection{A type-free alternative for Stage 1}
\label{sec:typefree}

Theorem~\ref{thm:final}'s Stage~1 search criterion, joint typicality
with a chosen type $W^\ell$, is defined
via that type directly. This is not very
satisfying: a construction whose entire point is to move past types
(Section~\ref{sec:motivation}) still routes its lookahead through one.
This subsection asks whether Stage~1 can look ahead to Stage~2
without types at all, and shows that the answer is yes.

For any $\hat{\bx}_1$, define --- the sequence-argument analogue of the
incremental Stage-2 rate $R^*(W^\ell)-R_1^*(W^\ell)$ a type induces ---
\begin{equation}
\rho_2^*(\hat{\bx}_1)\dfn-\frac1n\log V[\mathcal S(\bx,D_2)|\hat{\bx}_1],
\end{equation}
the actual conditional-LZ rate $\hat{\bx}_1$ induces for Stage~2 ---
built only from $U,V$, with no type anywhere. For any $\tau\ge0$, let
$\mathcal G_\tau \dfn \mathcal S(\bx,D_1)\cap\{\hat{\bx}_1 :
\rho_2^*(\hat{\bx}_1)\le\tau\}$.

\emph{Type-free schemes.} Call a two-stage scheme \emph{type-free} if
its search phases include no joint typicality criterion --- unlike
Theorem~\ref{thm:final}'s own Stage-1 search, which checks membership
in $\mathcal T(W_1^\ell|\bx)$ directly.

\begin{proposition}[Type-free achievability] \label{prop:typefree}
For every $\tau\ge0$ and $\epsilon>0$, and all sufficiently large $n$, there is a type-free,
two-stage LZ78-based scheme achieving, for every $\bx$,
\begin{equation}
L_1(\bx) \le -\log U[\mathcal G_\tau] + (2+\epsilon)\log n + c, \qquad
L_1(\bx)+L_2(\bx) \le -\log U[\mathcal G_\tau] + n\tau + (4+\epsilon)\log n + c,
\end{equation}
$c=c(\epsilon)$.
\end{proposition}
\begin{proof}
Stage~1 draws from the full, unrestricted $U$ and searches for
$\mathcal G_\tau$-membership; Lemma~\ref{lem:universal-ach} (stated
for any target set) applies directly, giving
$L_1(\bx)\le-\log U[\mathcal G_\tau]+(2+\epsilon)\log n+c$. The
resulting $\hat{\bx}_1\in \mathcal G_\tau$ satisfies
$\rho_2^*(\hat{\bx}_1)\le\tau$ by construction --- guaranteed, not
probabilistically. Stage~2 draws from the full $V(\cdot|\hat{\bx}_1)$
and searches for $\mathcal S(\bx,D_2)$ directly;
Lemma~\ref{lem:universal-ach} applies again, conditionally, giving
$L_2(\bx)\le n\rho_2^*(\hat{\bx}_1)+(2+\epsilon)\log n+c\le
n\tau+(2+\epsilon)\log n+c$, so
$L_1(\bx)+L_2(\bx)\le-\log U[\mathcal G_\tau]+n\tau+2(2+\epsilon)\log
n+2c$, which is the claimed bound after relabeling $2\epsilon,2c$ as
$\epsilon,c$.
\end{proof}

Dividing by $n$ gives
\begin{equation}
\rho_1(\bx) \le -\frac1n\log U[\mathcal G_\tau] + o(1), \qquad
\rho_1(\bx)+\rho_2(\bx) \le -\frac1n\log U[\mathcal G_\tau] + \tau + o(1),
\end{equation}
the form used below.

Ranging $\tau$ traces a family of achievable points, exactly
paralleling the type-indexed family $\mathcal W$ of
Section~\ref{sec:types}, but continuously and without reference to
any type.

\subsection{The type-free scheme matches the converse}
\label{sec:tf-converse}

Theorem~\ref{thm:conv-region} (Section~\ref{sec:types}) already
covers this scheme, like any other: it holds for every code,
whatever type its output happens to realize. What remains is only
the achievability side: a demonstration that
Proposition~\ref{prop:typefree}'s type-free scheme's own achieved
rate actually reaches, rather than merely respects, that converse.
The next theorem supplies this, by
reusing Lemma~\ref{lem:uniform-cov} (Section~\ref{sec:lz}) ---
used inside Theorem~\ref{thm:final} at only the one, Stage-1-selected
$\hat{\bx}_1$, it already holds for \emph{every}
$\hat{\bx}_1\in\mathcal T(W_1^\ell|\bx)$ simultaneously; that full
generality, unused until now, is exactly what is needed here.

\begin{theorem}[The type-free scheme reaches every point of $\mathcal
R(\bx)$] \label{thm:tf-match}
Recall Proposition~\ref{prop:typefree}'s construction: Stage~1 draws
from $U$ and searches for $\mathcal G_\tau$-membership --- jointly,
$D_1$-compatible \emph{and} with conditional-LZ rate at most $\tau$
--- and Stage~2 draws from $V(\cdot|\bxh_1)$ and searches directly
for $D_2$. For every $W^\ell\in\mathcal W$, running this construction
at
\begin{equation}
\tau(W^\ell)\dfn R^*(W^\ell)-R_1^*(W^\ell)+\Delta_n'(\ell)
\end{equation}
achieves
\begin{equation}
\rho_1(\bx)\;\le\;R_1^*(W^\ell)+\Delta_n(\ell)+o(1),\qquad
\rho_1(\bx)+\rho_2(\bx)\;\le\;R^*(W^\ell)+\Delta_n(\ell)+\Delta_n'(\ell)+o(1),
\end{equation}
converging to $(R_1^*(W^\ell),R^*(W^\ell))$ as $n\to\infty$. Ranging
$W^\ell$ over $\mathcal W$, the type-free scheme's own achievable
curve $\{(\rho_1(\tau),\rho_1(\tau)+\rho_2(\tau)):\tau\ge0\}$
therefore contains every corner of $\mathcal R(\bx)$ --- in
particular every point of its
Pareto frontier.
\end{theorem}
\begin{proof}
\emph{Step 1: the entire type-match set lies in
$\mathcal G_{\tau(W^\ell)}$.} Unpacking $\mathcal G_\tau$'s own
two-part definition, this means showing, for every
$\hat{\bx}_1\in\mathcal T(W_1^\ell|\bx)$: (i)
$\hat{\bx}_1\in\mathcal S(\bx,D_1)$, and (ii)
$\rho_2^*(\hat{\bx}_1)\le\tau(W^\ell)$.

(i) holds directly, by the type-determined-distortion property:
$\mathcal T(W_1^\ell|\bx)\subseteq \mathcal S(\bx,D_1)$.

For (ii): by Lemma~\ref{lem:uniform-cov}, every such $\hat{\bx}_1$
satisfies
\begin{equation}
V[\mathcal T(W^\ell|\bx,\hat{\bx}_1)|\hat{\bx}_1]\;\ge\;2^{-n[R^*(W^\ell)-R_1^*(W^\ell)]-n\Delta_n'(\ell)+o(n)}.
\end{equation}
Since $\mathcal T(W^\ell|\bx,\hat{\bx}_1)\subseteq \mathcal S(\bx,D_2)$
(the same type-determined-distortion property, now applied to
$W^\ell$'s own $(X,\hat X_2)$-marginal), $V[\mathcal S(\bx,D_2)|\hat{\bx}_1]$ is at least as large:
\begin{equation}
V[\mathcal S(\bx,D_2)|\hat{\bx}_1]\ge V[\mathcal
T(W^\ell|\bx,\hat{\bx}_1)|\hat{\bx}_1
]\ge 2^{-n[R^*(W^\ell)-R_1^*(W^\ell)]-n\Delta_n'(\ell)+o(n)}.
\end{equation}
Taking $-\frac1n\log$ of both sides (reversing the inequality) gives
exactly $\rho_2^*(\hat{\bx}_1)\le
R^*(W^\ell)-R_1^*(W^\ell)+\Delta_n'(\ell)+o(1)=\tau(W^\ell)+o(1)$.

Both (i) and (ii) hold for \emph{every} $\hat{\bx}_1\in\mathcal
T(W_1^\ell|\bx)$ simultaneously --- a uniform bound, not a
majority statement, since Lemma~\ref{lem:uniform-cov} itself already
holds uniformly. So $\mathcal T(W_1^\ell|\bx)\subseteq
\mathcal G_{\tau(W^\ell)+o(1)}$; absorbing this
vanishing term into $\tau(W^\ell)$'s own definition is harmless
(monotonicity of $\mathcal G_\tau$ in $\tau$), so $\mathcal
T(W_1^\ell|\bx)\subseteq \mathcal G_{\tau(W^\ell)}$ exactly, as used below.

\emph{Step 2: a lower bound on $U[\mathcal G_{\tau(W^\ell)}]$.} By
Lemma~\ref{lem:F0} applied to $\mathcal T(W_1^\ell|\bx)\subseteq
\mathcal T_n(Q_1^\ell)$, $U[\mathcal T(W_1^\ell|\bx)]\ge U_{Q_1^\ell}[\mathcal
T(W_1^\ell|\bx)]\cdot2^{-n\Delta_n(\ell)}=2^{-nR_1^*(W^\ell)-n\Delta_n(\ell)}$,
using $R_1^*(W^\ell)$'s own definition (Theorem~\ref{thm:conv}). By
Step~1 and monotonicity of $U[\cdot]$,
$U[\mathcal G_{\tau(W^\ell)}]\ge2^{-nR_1^*(W^\ell)-n\Delta_n(\ell)}$.

\emph{Step 3: convert to a rate bound.} Proposition~\ref{prop:typefree},
run at $\tau=\tau(W^\ell)$, gives
$\rho_1(\bx)\le-\frac1n\log U[\mathcal G_{\tau(W^\ell)}]+o(1)\le
R_1^*(W^\ell)+\Delta_n(\ell)+o(1)$, and
$\rho_1(\bx)+\rho_2(\bx)\le-\frac1n\log
U[\mathcal G_{\tau(W^\ell)}]+\tau(W^\ell)+o(1)\le
R_1^*(W^\ell)+\Delta_n(\ell)+[R^*(W^\ell)-R_1^*(W^\ell)+\Delta_n'(\ell)]+o(1)
=R^*(W^\ell)+\Delta_n(\ell)+\Delta_n'(\ell)+o(1)$, as claimed.
\end{proof}

Theorem~\ref{thm:tf-match}, combined with Theorem~\ref{thm:conv-region}
(Section~\ref{sec:types}, applying here exactly as it does to any
code), gives exact equality: the type-free scheme's own achievable
region --- the union, over $\tau\ge0$, of the quadrant above each
achieved $(\rho_1(\tau),\rho_1(\tau)+\rho_2(\tau))$ --- coincides
with $\mathcal R(\bx)$, not merely touches its corners from inside.
Neither direction needs to know, or control, which type $\mathcal
G_\tau$'s own random search actually realizes.

\emph{Search complexity: a real trade-off, not a wash.} This equality
is about achieved rates, not about what it costs to search for them,
and the two constructions are not equivalent on that second axis.
Theorem~\ref{thm:final}'s own, type-committed Stage-1 criterion --- joint
typicality of $(\bx,\hat{\bx}_1)$ with a type $W_1^\ell$ fixed in
advance --- is checkable from $\bx$ and a single candidate $\hat{\bx}_1$
alone, with no reference to Stage~2's own search space at all;
Stage~2 begins its own, separate search only once Stage~1 has
already succeeded, and only once. Proposition~\ref{prop:typefree}'s
type-free criterion, $\mathcal G_\tau$-membership, is not so cheap to
check: verifying it for a single candidate $\hat{\bx}_1$ requires
evaluating $\rho_2^*(\hat{\bx}_1)$, itself the $V$-measure of the entire
set $\mathcal S(\bx,D_2)$ conditioned on that one $\hat{\bx}_1$ --- a
computation reaching into Stage~2's own candidate space, repeated for
every $\hat{\bx}_1$ Stage~1 tries along the way, not merely the one it
eventually keeps. What the type-free construction buys for this
added cost is real, not merely conceptual: it needs no target type
$W^\ell$ fixed in advance, and so no prior solution of the type-side
optimization Section~\ref{sec:ach-region}'s own $\epsilon$-constraint
recipe requires just to decide which $W^\ell$ to aim for --- reaching
every point of $\mathcal R(\bx)$ off a single scalar dial $\tau$
instead; and, as the abstract states at the outset, no joint type of
$\ell$-vectors is ever constructed or checked by its own search
criterion, only individual-sequence quantities throughout.

\subsection{Domination over the finite-state-encoder route of
\cite{merhav244}} \label{sec:domination}

Section~\ref{sec:relation244} promised the two-stage extension of
\cite{merhav233}'s own single-stage domination inequality (its
eq.~(51)); this subsection proves it, now that $\mathcal G_\tau$, $U$, $V$,
and Proposition~\ref{prop:typefree} are in place. Recall
$\mathcal S_2(\bx,D_1,D_2)=\mathcal S(\bx,D_1)\times \mathcal S(\bx,D_2)$
(Section~\ref{sec:setup}) --- the admissible set exactly as
\cite{merhav244} defines it, since $d_1,d_2$ each depend on only their
own reconstruction, so the two constraints are independent.

\emph{The finite-state-encoder route's own achievable region.} For a genuine
$q$-state finite-state encoder --- $q$ bounded, not growing with $n$
--- \cite{merhav244}'s own scheme partitions each reconstruction into
non-overlapping blocks and restarts the LZ78 mechanism independently
within each block, trading some rate for a bounded state count
(unrestricted LZ78, run continuously over the whole sequence, needs a
number of states not small compared to $n$). Write $\mathcal
R_q(\bx)$ for the resulting true $q$-state achievable region.
Theorem~1 of \cite{merhav244} shows this is contained in the
region obtained from unrestricted LZ78 applied once, continuously,
over the whole sequence:
\begin{eqnarray}
\mathcal R_q(\bx) &\subseteq& \mathcal R^o(\bx) \dfn
\bigcup_{(\hat{\bx}_1,\hat{\bx}_2)\in\mathcal S_2(\bx,D_1,D_2)}
\Big\{(R_1,R_2): R_1\ge LZ(\hat{\bx}_1)/n-\Delta_1(q,n),\nonumber\\
&&\qquad\qquad R_1+R_2\ge \big[LZ(\hat{\bx}_1)+LZ(\hat{\bx}_2
|\hat{\bx}_1)\big]/n-\Delta_2(q,n)\Big\},
\end{eqnarray}
with $\Delta_1(q,n),\Delta_2(q,n)\to0$ as $q,n\to\infty$ (its
eqs.~(21)--(24)) --- the block-restart mechanism only shrinks the
achievable region further, it never grows it beyond $\mathcal
R^o(\bx)$.

In plain terms, this is what the next theorem shows: any rate pair
\cite{merhav244}'s own scheme achieves, via some particular
reconstruction $(\bxh_1,\bxh_2)$, is achieved by this paper's
type-free construction too, at least as well --- simply by adjusting
$\tau$ to match whichever $\bxh_1$ \cite{merhav244} would have used.
\cite{merhav244}'s own reconstruction is never reproduced, only the
threshold it implies.

\begin{theorem}[Two-stage domination over the finite-state-encoder
route] \label{thm:domination}
For any admissible reconstruction pair $(\bxh_1,\bxh_2)\in \mathcal
S_2(\bx,D_1,D_2)$ --- an arbitrary comparison target, entering only
through the threshold below, never as a search criterion --- setting
$\tau=\rho_2^*(\bxh_1)$ in Proposition~\ref{prop:typefree}'s
construction achieves
\begin{equation}
\rho_1(\bx)\;\le\;\frac1n LZ(\bxh_1)+o(1),\qquad
\rho_1(\bx)+\rho_2(\bx)\;\le\;\frac1n\big[LZ(\bxh_1)+LZ(\bxh_2|\bxh_1)\big]+o(1).
\end{equation}
\end{theorem}
\begin{proof}
Proposition~\ref{prop:typefree}'s own construction searches for
$\mathcal G_\tau$-membership without ever examining $\bxh_1$ itself
--- only the threshold $\tau$ depends on it. Since
$\bxh_1\in \mathcal S(\bx,D_1)$ and $\rho_2^*(\bxh_1)\le\tau$ (equality,
by choice of $\tau$), $\bxh_1\in \mathcal G_\tau$, so
$U[\mathcal G_\tau]\ge U(\bxh_1)\ge2^{-LZ(\bxh_1)}$ (Kraft's inequality
gives $U$'s own normalizing sum $Z\le1$, and $U(\bxh_1)=2^{-LZ(\bxh_1)}/Z\ge2^{-LZ(\bxh_1)}$). By
Proposition~\ref{prop:typefree}, $\rho_1(\bx)\le-\frac1n\log
U[\mathcal G_\tau]+o(1)\le\frac1n LZ(\bxh_1)+o(1)$. For the sum rate, the
same Kraft argument applied to the conditional measure
$V(\cdot|\bxh_1)$ --- itself a uniquely-decodable code given the
known side information $\bxh_1$ --- gives $V[\mathcal S(\bx,D_2)|\bxh_1]\ge
V(\bxh_2|\bxh_1)\ge2^{-LZ(\bxh_2|\bxh_1)}$, so
$\rho_2^*(\bxh_1)\le\frac1n LZ(\bxh_2|\bxh_1)$; combined with
Proposition~\ref{prop:typefree}'s $\rho_1(\bx)+\rho_2(\bx)\le
-\frac1n\log U[\mathcal G_\tau]+\tau+o(1)\le\frac1n LZ(\bxh_1)+\rho_2^*(\bxh_1)+o(1)$, this gives the claim.
\end{proof}

This is exactly \cite{merhav233}'s own single-stage inequality (its
eq.~(51), $\min_{\hat{\bx}\in \mathcal S(\bx,D)}LZ(\hat{\bx})\ge-\log U[\mathcal S(\bx,D)]$),
applied once unconditionally and once more conditionally on whichever
$\bxh_1$ is in play --- no chain rule for LZ complexity is needed,
since each inequality is proved separately, mechanically, from Kraft
alone. Since Theorem~\ref{thm:domination} holds for \emph{every}
admissible pair, not merely the best one, this paper's achievable
region contains, coordinate by coordinate, every rate pair in
$\mathcal R^o(\bx)$ --- and since $\mathcal
R_q(\bx)\subseteq\mathcal R^o(\bx)$, it contains
\cite{merhav244}'s own true achievable region too, up to
\cite{merhav244}'s own vanishing $(q,n\to\infty)$ slack, without
needing any separate analysis of the block-restart mechanism: this
paper's region is a superset of \cite{merhav244}'s, generally a
strict one. This proof leans on the fact that Ziv's
conditional LZ78 construction, built on $\bxh_1$ as side
information available to both encoder and decoder, is itself a
uniquely-decodable code satisfying its own
Kraft inequality for each fixed $\bxh_1$ --- confirmed directly by
Ziv's own construction \cite{ziv85}, which builds an explicit,
uniquely-decodable conditional LZ78 coding scheme achieving exactly
this length function; standard, and implicit throughout
\cite{merhav244}'s own use of the same construction.

\section{Extension to $r>2$ stages}
\label{sec:mstage}

The two-stage construction of Sections~\ref{sec:types}--\ref{sec:lz}
generalizes to any fixed number $r\ge2$ of stages along lines that
are natural given the two-stage case. What follows is an informal
description of this extension, not a formal proof: as emphasized
throughout, the two-stage case is where this paper's actual
difficulty lives, and the picture below is offered as the natural
next step, with no claim to the same rigor.

\emph{The type-covering skeleton.} Fix an $r$-way joint type $W^\ell$
of $(X,\hat X_1,\dots,\hat X_r)$ meeting all $r$ distortion budgets
$D_1,\dots,D_r$. Exactly as in the two-stage case (itself a direct
translation of \cite{rimoldi94}'s classical characterization to
$\ell$-blocks), the type-conditional bound (Theorem~\ref{thm:conv}'s own
Kraft-counting argument, run once per prefix $k=1,\dots,r$, over the
$k$-tupled alphabet $\Xh_1\times\cdots\times\Xh_k$, each stage's own,
not assumed equal) and the achievability construction
(Theorem~\ref{thm:ach}'s own Stage-1/Stage-2 mechanism, run once per
stage, each conditioning on the exact history realized so far)
generalize with no new idea: at stage $k$, search for the first
candidate $\hat{\bx}_k$ completing the full joint type $W^\ell$
restricted to $(\bx,\hat{\bx}_1,\dots,\hat{\bx}_k)$, given the exact
$(\hat{\bx}_1,\dots,\hat{\bx}_{k-1})$ already fixed by the previous $k-1$
stages. This gives $\rho_1(\bx)+\dots+\rho_k(\bx)\to R_k^*(W^\ell)$ for every $k$
simultaneously, $R_k^*(W^\ell)$ the natural generalization of $R_1^*,R^*$
above to the cumulative rate through stage $k$.

\emph{The LZ-based realization.} At stage $k$, candidates for
$\hat{\bx}_k$ would be drawn i.i.d.\ from the measure proportional to
$2^{-LZ(\hat{\bx}_k|\hat{\bx}_1,\dots,\hat{\bx}_{k-1})}$ --- the conditional
LZ78 codelength of $\hat{\bx}_k$ given all $k-1$ previously
realized reconstructions as side information, generalizing
$U$ (stage~1, $k=1$, no prior history) and $V(\cdot|\bxh_1)$
(stage~2, conditional on the single prior reconstruction) to
conditioning on the full history. The same principle that resolves
the two-stage Jensen obstruction (Section~\ref{sec:types}) would
need to apply at every stage: $\hat{\bx}_k$'s search criterion must
guarantee, not merely make likely, that $\hat{\bx}_k$ keeps stage
$k+1$'s own rate low --- generalizing Stage~1's own joint-typicality
criterion
from looking one stage ahead to looking ahead across however many
stages remain, by requiring joint typicality of the full history
$(\bx,\hat{\bx}_1,\dots,\hat{\bx}_k)$ rather than of $(\bx,\hat{\bx}_1)$
alone.

\emph{The type-free realization: full lookahead by backward
recursion.} Sections~\ref{sec:typefree}--\ref{sec:tf-converse}'s
type-free scheme also generalizes, and settles the lookahead-depth
question just raised: one-step lookahead (stage $k$ guaranteeing only
stage $k+1$'s affordability) does \emph{not} suffice, because stage
$k+1$ itself must guarantee stage $k+2$'s affordability, so the
quantity stage $k$ needs to bound is stage $k+1$'s own
\emph{restricted} rate, not its naive, unrestricted one --- a
dependency running the wrong way for one-step-at-a-time chaining.
Full lookahead across the entire remaining chain is genuinely needed,
but it is a single \emph{backward} recursion, computed once as
design-time bookkeeping, not a search the running encoder performs.
Write $\hat{\bx}^k\dfn(\hat{\bx}_1,\dots,\hat{\bx}_k)$ for a history of length
$k$ ($\hat{\bx}^0$ empty), and $U_k(\cdot|\hat{\bx}^{k-1})$ for the
measure proportional to $2^{-LZ(\cdot|\hat{\bx}_1,\dots,\hat{
\bx}_{k-1})}$ just described ($U_1\dfn U$, no history; $U_2(\cdot|
\hat{\bx}_1)\dfn V(\cdot|\hat{\bx}_1)$, recovering Stage~2's own
measure). Concretely: define $\rho_r^*(\hat{\bx}^{r-1})\dfn-\frac1n\log
U_r[\mathcal S(\bx,D_r)|\hat{\bx}^{r-1}]$ (unrestricted --- nothing follows
stage $r$), and recursively, for $k=r-1,\dots,1$, given
$\tau_{k+1},\dots,\tau_{r-1}$ already fixed,
\begin{equation}
\mathcal G^{(k)}_{\tau_k}(\hat{\bx}^{k-1})\dfn \mathcal
S(\bx,D_k)\cap\{\hat{\bx}_k:
\rho^{*\to}_{k+1}(\hat{\bx}^{k-1},\hat{\bx}_k)\le\tau_k\},\qquad
\rho^{*\to}_{k+1}(\hat{\bx}^k)\dfn-\frac1n\log
U_{k+1}\big[\mathcal G^{(k+1)}_{\tau_{k+1}}(\hat{\bx}^k)|\hat{\bx}^k\big],
\end{equation}
generalizing $\mathcal G_\tau$ ($k=1$ case, once $\tau_1$ is fixed as below)
to every stage. Since $\hat{\bx}_1\in \mathcal G^{(1)}_{\tau_1}$ or $\hat{
\bx}_k\in \mathcal G^{(k)}_{\tau_k}(\hat{\bx}^{k-1})$ are, by this definition,
statements about the specific, already-realized $\hat{\bx}_k$ ---
checked before acceptance, not averaged over how it arose --- the
guarantee never has to survive the randomness of how the previous
stage's search turned out; this is the same mechanism that resolves
the $r=2$ Jensen obstruction, applied once per stage rather than
once. A downward induction (from $k=r-1$ to $k=1$) shows the entire
type-match set $\mathcal T(W_{\le k}^\ell|\bx)$ is absorbed into
$\mathcal G^{(k)}_{\tau_k(W^\ell)}$ once
$\tau_k(W^\ell)\dfn R_{k+1}^*(W^\ell)-R_k^*(W^\ell)+\Delta_n^{(k)}(\ell)$
--- $\Delta_n^{(k)}(\ell)\to0$ the natural, conditioning-depth-$k$
generalization of $\Delta_n'(\ell)$ (itself the depth-$1$ case) ---
mirroring Lemma~\ref{lem:uniform-cov} at every conditioning depth; the
resulting rates telescope exactly,
$\rho_1(\bx)+\dots+\rho_k(\bx)\to R_k^*(W^\ell)$ for every $k$
simultaneously, matching Theorem~\ref{thm:tf-match}'s $r=2$ result at
every stage, not merely at the last one.

\emph{What remains.} The lookahead-depth question raised above is
settled by the construction itself --- full lookahead, via backward
recursion, with no obstruction of a new character, the same
mechanism as $r=2$ repeated $r-1$ times rather than found anew. One
place this construction might look, at first, like it leans on
something beyond what the two-stage case already uses: Stage $k$'s
search needs Lemma~\ref{lem:condF0}'s codelength bound with the
side information taken jointly from $k-1$ prior reconstructions, not
the single $\hat{\bx}_1$ the lemma is stated for. But conditioning on a
fixed-length tuple $(\hat{\bx}_1,\dots,\hat{\bx}_{k-1})$ is conditioning on
\emph{one} sequence over the product alphabet $\hat
X_1\times\dots\times\hat X_{k-1}$ --- finite, since each factor is
finite and $k-1\le r-1$ is bounded independent of $n$ --- and nothing
in Lemma~\ref{lem:condF0}'s own derivation uses any structural
property of the conditioning alphabet beyond its being finite. So
this is the same lemma, applied once with the conditioning
alphabet relabeled, not a separate instance needed at each depth ---
no additional assumption, at any $k$, beyond the one the $r=2$ case
already carries.

\section{Summary and Conclusion}
\label{sec:discussion}

This paper extended the single-stage, LZ-based random-coding
approach of \cite{merhav233} to two-stage successive refinement of
individual sequences. Section~\ref{sec:types} established a
complete direct/converse match in the language of type classes ---
a two-stage analogue of Rimoldi's classical characterization ---
culminating in a converse (Theorem~\ref{thm:conv-region}) that binds
every code, not merely type-committed ones. Section~\ref{sec:lz}
realized the same rates via an LZ78-based random-coding scheme,
universal across source statistics and distortion measures, with no
block-length parameter built into the codebook itself.
Sections~\ref{sec:typefree}--\ref{sec:tf-converse} then gave a fully
type-free construction, reaching every point of the achievable
region without ever committing to a type --- matching the same
converse exactly, though at a real search-cost premium over the
type-committed route.

This type-free flexibility turns out not to enlarge the achievable
rate region itself: the type-free and type-committed constructions reach
exactly the same points. The gain is elsewhere.
Section~\ref{sec:domination} showed the type-free construction
dominates the finite-state-encoder approach of \cite{merhav244}: any
rate pair that scheme achieves is achieved by this paper's own
construction too, simply by adjusting the search threshold
accordingly --- extending \cite{merhav233}'s own single-stage
domination result to two stages. Section~\ref{sec:mstage} extended
the type-free construction to a general, fixed number $r>2$ of
stages, via a single backward recursion, with no obstruction of a
new character beyond what the two-stage case already resolves.

\appendix
\setcounter{equation}{0}
\renewcommand{\theequation}{A.\arabic{equation}}
\section{Proof of the universal random-coding achievability lemma}
\label{app:proofs}

This proof is closely modeled on \cite{merhav233}'s
own achievability construction (its Theorem~2).

\begin{proof}[Proof of Lemma~\ref{lem:universal-ach} (Universal random-coding achievability)]
This follows the same construction as \cite{merhav233}'s own
achievability theorem (its Theorem~2), generalized in two respects:
the target, from \cite{merhav233}'s own specific family of
distortion balls $\mathcal S(\bx,D)$ to an arbitrary target family
$\mathcal E(\bx)$; and the measure, from the specific $U$ to an
arbitrary $\mu$ meeting the stated hypothesis ($\mu(\bxh)\ge
B^{-n}$ for every $\bxh$ with $\mu(\bxh)>0$). Two further
differences below are purely mechanical: the expectation bound uses
the untruncated geometric mean directly, where \cite{merhav233}
carries the $A^n$ truncation through an explicit sum to the same
conclusion; and the worst-case LZ bound needed for search failure is
drawn from Lemma~\ref{lem:F0}, already established earlier in this
paper, rather than re-derived from the counting-sequence
construction \cite{merhav233} uses for the same purpose.

\emph{Index encoding.} Given the codebook $\hat{\bx}_1,\dots,\hat{\bx}_{A^n}$
(drawn once, i.i.d.\ $\sim\mu$, then revealed to encoder and decoder
alike) and a source $\bx$ with target $\mathcal E(\bx)$, let $I(\bx)$ be the
index of the first codeword landing in $\mathcal E(\bx)$, or $I(\bx)=A^n$ if
none does. Encode $I(\bx)$ using the length function
$L(\bx)=-\log u[I(\bx)]$ for $u[i]=(1/i)/\sum_{k=1}^{A^n}(1/k)$,
$i=1,\dots,A^n$ --- a genuine probability distribution on
$\{1,\dots,A^n\}$, so $L$ is realizable by a valid prefix-free code.
Since $\sum_{k=1}^{A^n}1/k\le\ln(A^n)+1=n\ln A+1$,
\begin{equation}
L(\bx) \le \log I(\bx) + \log(n\ln A+1) \le \log I(\bx)+\log n+c, \qquad c\dfn\log(\ln A+1).
\end{equation}

\emph{A single-target expectation bound.} For fixed $\bx$, if all
codewords are drawn i.i.d.\ $\sim U$, then for any positive integer
$N$, $\Pr[I(\bx)>N]=(1-U[\mathcal E(\bx)])^N\le e^{-N\cdot U[\mathcal E(\bx)]}$. By
Jensen's inequality (once), writing $p\dfn U[\mathcal E(\bx)]$,
\begin{equation}
E[\log I(\bx)] \le \log E[I(\bx)] = \log(1/p) = -\log(U[\mathcal E(\bx)]),
\end{equation}
using the standard geometric-distribution fact $E[I]=1/p$ (treating
$I(\bx)$, before the truncation at $A^n$, as the index of the first
success among i.i.d.\ Bernoulli($p$) trials). Combining with the
encoding bound above,
\begin{equation}
E[L(\bx)] \le -\log(U[\mathcal E(\bx)])+\log n+c \dfn L^+(\bx).
\end{equation}

\emph{From one $\bx$ in expectation to every $\bx$ simultaneously.}
Define, exactly as in \cite{merhav233}'s own eq.~(42) --- with the
excess threshold $(1+\epsilon)\log n$, not merely $\epsilon\log n$;
this distinction matters and is explained below ---
\begin{equation}
E_n \dfn E\Big\{\max\Big(\max_{\bx\in\mathcal X^n}\bone\{I(\bx)=A^n\mbox{ fails}\},\ \big[\max_{\bx\in\mathcal X^n}\big(L(\bx)-L^+(\bx)-(1+\epsilon)\log n\big)\big]_+\Big)\Big\},
\end{equation}
the expectation over the random codebook. If $E_n\to0$, then for
some $N=N(\epsilon,A,B)$ and every $n>N$, some codebook realization
makes both terms inside the
max simultaneously below $1$ for every $\bx$ --- the first term is
$0$ or $1$-valued, so below $1$ means the search succeeds for every
$\bx$, and the second gives $L(\bx)\le L^+(\bx)+(1+\epsilon)\log n =
-\log(U[\mathcal E(\bx)])+(2+\epsilon)\log n+c$ for
every $\bx$, which is the claimed bound. It remains to show $E_n\to0$;
since the max of two nonnegative terms is at most their sum, bound
each separately.

\emph{Search failure, uniformly.} By the union bound over $\bx$ and
$\Pr[I(\bx)=A^n\mbox{ fails}]=(1-p)^{A^n}\le e^{-A^np}$,
\begin{equation}
E\Big\{\max_{\bx}\bone\{\mbox{failure}\}\Big\} \le \sum_{\bx\in\mathcal X^n}e^{-A^nU[\mathcal E(\bx)]}.
\end{equation}
Since $U[\mathcal E(\bx)]\ge U(\hat{\bx}_0)$ for any single $\hat{\bx}_0\in
\mathcal E(\bx)$ (assuming $\mathcal E(\bx)\ne\emptyset$, else the achievability claim
is vacuous), and $U$'s normalization together with the LZ78
codelength bound (Lemma~\ref{lem:F0},
$LZ(\hat{\bx})\le n(1+\epsilon_n)\log \hat K$ for the worst case, exactly
as in \cite{merhav233}) gives $U[\mathcal E(\bx)]\ge2^{-n(1+\epsilon_n)\log \hat K}$
uniformly, this sum is at most $|\mathcal
X|^n\exp\{-2^{n[\log A-(1+\epsilon_n)\log \hat K]}\}\to0$ doubly
exponentially fast whenever $A>\hat K$, exactly as in
\cite{merhav233}'s own Theorem~2.

\emph{Excess codelength, uniformly --- why the threshold must be
$(1+\epsilon)\log n$, not $\epsilon\log n$.} As in \cite{merhav233},
\begin{equation}
E\Big\{\big[\max_{\bx}(L(\bx)-L^+(\bx)-(1+\epsilon)\log n)\big]_+\Big\}
\le\sum_{\bx\in\mathcal X^n}\int_0^{n\log A}\Pr\Big[I(\bx)\ge\frac{2^{(1+\epsilon)\log n+s}}{U[\mathcal E(\bx)]}\Big]\,\mathrm ds
\end{equation}
\begin{equation}
\le\sum_{\bx\in\mathcal X^n}\int_0^{n\log A}\exp\{-2^s n^{1+\epsilon}\}\,\mathrm ds
\le|\mathcal X|^n(n\log A)e^{-n^{1+\epsilon}}\to0,
\end{equation}
using the geometric tail bound $\Pr[I(\bx)\ge N]\le
e^{-Np}\le\exp\{-2^{(1+\epsilon)\log n+s}\cdot U[\mathcal E(\bx)]\}$ with
$N=2^{(1+\epsilon)\log n+s}/U[\mathcal E(\bx)]$, evaluated at $s=0$ for the
final step. This is precisely why the threshold cannot merely be
$\epsilon\log n$: that would give $|\mathcal
X|^n\exp\{-n^\epsilon\}$ in place of $|\mathcal
X|^n\exp\{-n^{1+\epsilon}\}$, and $n^\epsilon$ grows \emph{slower}
than $n\log|\mathcal X|$ for $\epsilon<1$, so this weaker bound
\emph{diverges} rather than vanishing --- the polynomial threshold
$n^{1+\epsilon}$, not just $n^\epsilon$, is what is needed to
dominate the $|\mathcal X|^n$ union-bound factor. Both terms
$\to0$, so $E_n\to0$, completing the proof --- but only for
$n>N(\epsilon)$: it is exactly this term,
$|\mathcal X|^n(n\log A)e^{-n^{1+\epsilon}}$, whose rate of decay
sets $N(\epsilon)$, since smaller $\epsilon$ needs larger $n$ before
$n^{1+\epsilon}$ overtakes the union-bound factor $|\mathcal X|^n$,
so $N(\epsilon)\to\infty$ as $\epsilon\to0^+$; at $\epsilon=0$
exactly this same term becomes $|\mathcal X|^n(n\log A)e^{-n}$,
which diverges rather than vanishes whenever $|\mathcal X|\ge3$, and
the argument breaks down outright.
\end{proof}

\section{Proof of the Type-Conditional Rate Bound}
\label{app:thm1}
\setcounter{equation}{0}
\renewcommand{\theequation}{B.\arabic{equation}}

The proof uses that any single reconstruction can be
responsible for at most
$|\mathcal T_n(P^\ell)\cap\mathcal T(W_1^\ell|\bz)|$ sources of type
$W_1^\ell$, itself already given by
Lemma~\ref{lem:dc}'s type-match extension, to convert a bound on the
fraction of \emph{codewords} into one on the fraction of
\emph{source sequences} directly. This observation is used
throughout the proof that follows.

\begin{proof}[Proof of Theorem~\ref{thm:conv} (Type-conditional rate bound)]
Fix the type $W_1^\ell$; the argument for $W^\ell$ (total rate) is
identical with $W_1^\ell$ replaced by $W^\ell$, $Q_1^\ell$ by
$W^\ell$'s own $(\hat X_1,\hat X_2)$-marginal, and the reconstruction
alphabet $\Xh_1$ by the doubled alphabet $\Xh_1\times\Xh_2$ (treating the
pair $(\bxh_1,\bxh_2)$ as a single reconstruction), so we give the
single argument and combine at the end.

\emph{A codeword-counting bound.} Throughout, $\bz$ denotes a generic
candidate ranging over $\mathcal T_n(Q_1^\ell)$, kept distinct from
$\bxh_1$, reserved below for $\Phi$'s own, specific output on a
given $\bx$. By Lemma~\ref{lem:dc}'s type-match
extension, writing $\mathcal T(W_1^\ell|\bz)\dfn\{\bx:
\hat P_{\bx\bz}=W_1^\ell\}$ for the reverse-direction
type-match set, $|\mathcal T_n(P^\ell)\cap\mathcal T(W_1^\ell|\bz)|=|\mathcal T_n(P^\ell)|\cdot
U_{Q_1^\ell}[\mathcal T(W_1^\ell|\bx)]$ for every $\bz\in
\mathcal T_n(Q_1^\ell)$ and every $\bx\in \mathcal T_n(P^\ell)$; write $u_1\dfn
U_{Q_1^\ell}[\mathcal T(W_1^\ell|\bx)]$ for this common value (a
single number, the same for every choice of $\bz$ and $\bx$ in their
respective type classes, by Lemma~\ref{lem:dc}'s own symmetry). Fix
$\ell_0>0$, and let $\mathcal Z_{\ell_0}\dfn\{\bz:
\bz=\psi_1(\phi_1(\bx))\mbox{ for some }\bx\mbox{ with }\hat
P_{\bx\bz}=W_1^\ell\mbox{ and }|\phi_1(\bx)|<\ell_0\}$ --- $\bx\in
\mathcal T_n(P^\ell)$ and $\bz\in \mathcal T_n(Q_1^\ell)$ are automatic,
already implied by $\hat P_{\bx,\bz}=W_1^\ell$ --- the set of
type-$Q_1^\ell$
reconstructions occurring as outputs of type-$W_1^\ell$
sources with codeword length below $\ell_0$. The claim follows from
four statements.

\emph{(i) $|\mathcal Z_{\ell_0}|<2^{\ell_0}$.} The first-stage decoder $\phi_1$ assigns each
distinct reconstruction its own binary string,
so $\mathcal Z_{\ell_0}$ injects into the binary
strings of length $<\ell_0$, of which there are
$\sum_{k<\ell_0}2^k=2^{\ell_0}-1$.

\emph{(ii) For every $\bz\in \mathcal T_n(Q_1^\ell)$,
$|\{\bx\in \mathcal T_n(P^\ell):\hat P_{\bx\bz}=W_1^\ell,\
\psi_1(\phi_1(\bx))=\bz\}|\le|\mathcal T_n(P^\ell)|\cdot u_1$ ($u_1\dfn
U_{Q_1^\ell}[\mathcal T(W_1^\ell|\bx)]$, as above).} Every such
$\bx$ satisfies $\hat P_{\bx\bz}=W_1^\ell$ by definition of
the set being counted, so this set is contained in
$\mathcal T_n(P^\ell)\cap\mathcal T(W_1^\ell|\bz)$, of size
$|\mathcal T_n(P^\ell)|\cdot u_1$ by the opening display.

\emph{(iii) The sets $\{\bx\in \mathcal T_n(P^\ell):\psi_1(\phi_1(\bx))=\bz\}$,
$\bz\in\mathcal Z_{\ell_0}$, are pairwise disjoint.} $\psi_1\circ\phi_1$
is a function of $\bx$, so each $\bx$ belongs to exactly one such set.

\emph{(iv) Combining (i)--(iii):}
\begin{eqnarray}
|\{\bx\in \mathcal T_n(P^\ell): (\bx,\bxh_1)\mbox{ has type }W_1^\ell,\ L_1(\bx)<\ell_0\}|
&=& \sum_{\bz\in\mathcal Z_{\ell_0}}|\{\bx:\hat P_{\bx\bz}=W_1^\ell,\ \psi_1(\phi_1(\bx))=\bz\}|\nonumber\\
&\le&|\mathcal Z_{\ell_0}|\cdot|\mathcal T_n(P^\ell)|\cdot u_1
\ <\ 2^{\ell_0}|\mathcal T_n(P^\ell)|u_1,
\end{eqnarray}
the first equality since a union of pairwise disjoint sets (iii) has
size equal to the sum of their sizes, the first inequality by (ii)
applied to each term, and the last by (i).

Choosing $\ell_0=\log_2(1/u_1)-\epsilon\log_2 n$, so
$2^{\ell_0}=n^{-\epsilon}/u_1$, this bound becomes exactly
$n^{-\epsilon}|\mathcal T_n(P^\ell)|$: at most an $n^{-\epsilon}$ fraction of
$\bx\in \mathcal T_n(P^\ell)$ can have output type $W_1^\ell$ and
$L_1(\bx)<-\log_2(u_1)-\epsilon\log_2 n = nR_1^*(W^\ell)-\epsilon\log n$.
Equivalently, for all but at most an $n^{-\epsilon}$ fraction of
$\bx\in \mathcal T_n(P^\ell)$, every code whose output on this $\bx$ has type
$W_1^\ell$ satisfies the first bound.

The identical argument on the doubled alphabet (reconstruction
$(\bxh_1,\bxh_2)$, type $W^\ell$, target $\mathcal T(W^\ell|\bx)$,
total codeword length $L_1+L_2$, since $(\phi_1,\phi_2)$ jointly
determine one binary string of this length from which
$(\bxh_1,\bxh_2)$ is recovered) excludes a second set of at most an
$n^{-\epsilon}$ fraction of $\bx\in \mathcal T_n(P^\ell)$. By the union bound,
all but at most a $2n^{-\epsilon}$ fraction of $\bx\in \mathcal T_n(P^\ell)$
avoid both exceptional sets, giving both bounds simultaneously
whenever the code's output on such an $\bx$ has joint type $W^\ell$
(which forces its own $(X,\hat X_1)$ sub-type to be $W_1^\ell$, so
the first bound is exactly the marginal case of the second).
\end{proof}

\end{document}